\documentclass[11pt,english]{article}
\usepackage{lipsum}

\usepackage{lmodern}

\usepackage[T1]{fontenc}
\usepackage[latin9]{inputenc}
\usepackage[margin=0.9in]{geometry}
\usepackage{xcolor}
\usepackage{etoolbox}
\usepackage{array}
\usepackage{float}
\usepackage{amsmath}
\usepackage{amssymb}
\usepackage[flushleft]{threeparttable}
\usepackage{dcolumn}
\usepackage{setspace}
\usepackage{pdflscape}
\usepackage{hyperref}
\usepackage{booktabs,threeparttable}
\usepackage[authoryear]{natbib}
\PassOptionsToPackage{normalem}{ulem}
\usepackage{ulem}
\usepackage{comment}
\makeatletter
\usepackage{fancyvrb}

\usepackage{xurl}
\usepackage[toc,page]{appendix}

\usepackage[english]{babel}

\usepackage{amsmath}
\usepackage{amsthm}
\usepackage{amssymb}
\usepackage{upgreek}

\usepackage{mathtools}

\newcommand{\sign}{\text{sign}}

\usepackage{chngcntr}
\usepackage{apptools}
\AtAppendix{\counterwithin{lemma}{section}}

\newtheorem{lemma}{Lemma}
\newtheorem{proposition}{Proposition}

\usepackage{color}

\usepackage{caption}
\date{}

\usepackage{babel}

\usepackage{graphicx}
\usepackage{caption}
\usepackage{subcaption}
\usepackage{titling}
\usepackage{tikz}
\usetikzlibrary{arrows,calc, patterns, positioning, shapes.geometric, decorations.pathreplacing,decorations.markings}
\usepackage{enumitem}

\usepackage{multirow}

\begin{document}
\begin{titlingpage}
\setstretch{1.3}

\title{Social Media, Content Moderation, and Technology}

\author{Yi Liu, T. Pinar Yildirim, Z. John Zhang \thanks{Liu is a doctoral student at the University of Pennsylvania's Wharton School, Marketing Department. Yildirim is an Assistant Professor and Zhang is the Tsai Wan-Tsai Professor at the same department. 
All correspondence can be sent to the second author. All errors are our own. }}

\vspace{2in}

\date{\today}

\maketitle





\begin{abstract}

This paper develops a theoretical model to study the economic incentives for a social media platform to moderate user-generated content. We show that a self-interested platform can use content moderation as an effective marketing tool to expand its installed user base, to increase the utility of its users, and to achieve its positioning as a moderate or extreme content platform. The optimal content moderation strategy differs for platforms with different revenue models, advertising or subscription. 
We also show that a platform's content moderation strategy depends on its technical sophistication. Because of imperfect technology, a platform may optimally throw away the moderate content more than the extreme content. Therefore, one cannot judge how extreme a platform is by just looking at its content moderation strategy. 
Furthermore, we show that a platform under advertising does not necessarily benefit from a better technology for content moderation, but one under subscription does. This means that platforms under different revenue models can have different incentives to improve their content moderation technology. Finally, we draw managerial and policy implications from our insights.

\bigskip
\noindent {\bf keywords:} social media platforms, content moderation, revenue models, technology
\end{abstract}
\end{titlingpage}

\setstretch{1.4}
\newpage 
\section{Introduction}









A significant challenge that online social media platforms such as Facebook and Twitter face today is acting as the custodians of the Internet while at the same time being the  center of self-expression and user-generated content \citep{gillespie2018custodians}. Social media platforms allow millions of users with diverse views to post their opinions on issues day-to-day, some of which are deemed offensive, harmful, or ``extreme\footnote{Through the rest of the paper, we will refer to such content as ``extreme content.''},'' by majority of users. Users demand, on the one hand, to freely express their views on ongoing political, social, and economic issues on social media platforms without intervention and without being told their views are ``inappropriate.'' On the other hand, they abhor the content that they themselves view as inappropriate, sensitive, harmful, or extreme. So platforms, in some form or another, moderate content to protect individual users and their interests, by removing posts users may consider extreme. In fact, some executives at Facebook view content moderation as ``the most important thing they do'' \citep{Fbcontentm}. 
In 2019, Facebook CEO Mark Zuckerberg declared that they would be allocating 5\% of the firm revenues, \$3.7 billion, on content moderation, an amount greater than Twitter's entire annual revenue \citep{MarkZuckerbergSays}. 
In this paper, we take a theoretical look at how a self-interested platform can do ``the most important thing.''

Content moderation is no simple feat. \cite{Bigtechneedsmore} states that ``platforms like Facebook have to make trade-offs ... between free expression and safety'' and that there is rarely a clear ``right'' answer. According to a {\em Morningconsult.com} survey\footnote{Source: \url{https://morningconsult.com/wp-content/uploads/2019/10/190859_crosstabs_CONTENT_MODERATION_Adults_v4_JB-1.pdf}. The survey also states that 56\% of adults think that edited or distorted images of public officials and celebrities should be removed by social media sites. The same number is 69\% for ``Misleading health information,'' 32\% for ``Fad diets, such as detoxes.''}, consumers vary on their tolerance to potentially harmful content. Of those surveyed, 80\%  wish to see hate speech such as posts using slurs against a racial, religious or gender group removed, 73\% wish to see videos depicting violent crimes removed, and only 66\% wish to see depictions of sexual acts removed. This user heterogeneity adds complexity to content moderation, but it also gives a social media platform the cover and leeway to conduct content moderation to achieve its own profit objective. This is because the balance between self-expression and safety in the context of user heterogeneity can justify any strict or lax content moderation strategy motivated by a platform's profit.  In this study, we incorporate the user heterogeneity and derive a platform's optimal content moderation strategy.


For social media platforms, moderation determines which content is removed, therefore what users can post and read. Since the basic product offered by social media platforms is user-generated content, removal of some of these content simultaneously changes the design of the product offered to other users, and endogenously determines which users may enjoy the content enough to stay on the social media as well. Put differently, content moderation for a social media platform is a decision that simultaneously determines the content offered and the platform's positioning as a moderate or extreme content platform.  It is a decision that combines ``product'' and ``promotion'' in one, and it is also a decision that attracts extensive public scrutiny and regulatory attention. How exactly content moderation should be implemented is also an issue that is of high priority to policy makers, academics, and industry pundits \citep{jhaver2018online,BIcontentmoderationreport,Bidentechadvisor}, and there is a heated  debate  ongoing on the topic. 

Given the importance of content moderation to a platform and to the public, it is important to understand this big and complex topic from ground up starting with  how a self-interested platform may conduct content moderation.  In this paper, we take a first step toward understanding the economic foundations of a platform's desire to moderate content and to invest in content moderation technology. We study three key questions that are at the heart of marketing and management for social media platforms. First, how do self-interested social media platforms moderate content given their revenue model? More specifically, does the revenue source (advertising vs. subscription) matter in their motivation and strategy to conduct content moderation? Second, does a  platform with a given revenue model  always prefer   a  better  technology for  content  moderation, so that they have sufficient incentives to pursue  the best  technology on their own? Finally, how does a free-market content moderation compare to the social optimum? The answers to these questions are important not only to understand a platform's behaviors with regard to content moderation, but also to identify the rationale for any regulatory intervention or non-intervention.

In this paper, we develop a theoretical model to address all these questions. 
In our model, a  platform allows users to post and share content with others, and earns revenue either through advertising or subscription fees. Users enjoy the ability to express their opinions, and read others' content, from which they may or may not obtain positive utility depending on their own preferences and also on how extreme the content on the platform is.  
A platform moderates online content to maximize its revenue. This model setup allows us to study content moderation as a marketing tool and explore its strategy and public policy implications.

The analysis of our model shows that content moderation by a platform is primarily motivated by users' preferences for posting vs. reading content on the platform. A platform conducts content moderation only if users care more about reading others' content than  posting their own. As a marketing tool, content moderation can  perform two functions for a platform: expand its user base and increase users' willingness-to-pay. The user base is expanded through pruning extreme users to get more users of moderate opinions. Content moderation can also increase users' willingness-to-pay by reducing their reading disutility from extreme content. This result establishes content moderation as an effective product design and positioning tool.

When applying this tool, a platform's optimal content moderation strategy will depend on its revenue model and also its technological sophistication. When a platform chooses optimally not to conduct content moderation, the platform under advertising will field a less extreme content than a platform under subscription, all else being equal. When a platform optimally conducts content moderation, its content would be less extreme under subscription than under advertising. Furthermore, because of imperfect technology, the optimal content moderation strategy may call for a platform to prune the moderate content more than the extreme content and vice versa depending on the revenue model. This analysis thus suggests that the content moderation strategy for a platform may not be as straightforward as removing the most extreme content while keeping all the moderate content.

Technology plays an important role in content moderation also in a different way. When criticized for insufficient effort at content moderation, social media executives frequently blame imperfect technology  and promise to remedy the inadequacy through technology improvement \citep{SocialMediagiantswarnofAI,timelineforAIto,AIenhancementshavebolstered}.  Interestingly, our analysis shows that for a self-interested social media platform, technological improvement does not always lead to more content moderation or to  less extreme content on the platform. In addition, a platform under advertising may not even benefit from a better technology. In other words, a social media platform under advertising may not have the incentive to perfect its technology for content moderation. 
This result demonstrates that content moderation on online platforms is not merely an outcome of their technological capabilities, but economic incentives. Indeed, even the technology constraint for extreme content identification itself can be a strategic choice for the platform.
This result thus casts some doubts on whether social media platforms will always remedy the technological deficiencies on their own.

The regulatory concerns are even deeper when one compares the content moderation strategy for a self-interested platform with that for a social planner. We show that a social planner  will use content moderation to prune the users whose net utility contribution to the society is negative. In addition, the social planner always pursues perfect technology if the cost of developing technology is not an issue. In contrast, a self-interested platform under either advertising or subscription is always more likely to conduct content moderation than a social planner, and when conducting content moderation, a platform under advertising (subscription) will be less (more) strict than a social planner. Only a platform under subscription will have an interest aligned with a social planner in perfecting the technology for content moderation. These conclusions thus demonstrate that there is room for government regulations and when they are warranted, they need to be differentiated with regard to the revenue model a platform adopts.


Studies in the past on user-generated content (UGC), social media, and firm strategy have taken a number of different directions. Many studies have looked into the dynamics of and the motivations for  user-generated content \citep[e.g.,][]{toubia2013intrinsic, daugherty2008exploring,sun2017motivation,iyer2016competing,ahn2016managing,bazarova2014self,buechel2015motivations}. A number of empirical and theoretical studies have also investigated how firms can glean information from UGC and  use  it strategically to perform their marketing functions \citep[e.g.,][]{ghose2012designing,timoshenko2019identifying,iyengar2011opinion,tirunillai2012does,goh2013social,godes2004using}. As the UGC provides a different dimension for firms' offerings, a number of theoretical papers have derived the optimal differentiation strategy for competing firms \citep[e.g.,][]{yildirim2013user,zhang2015differentiation}. However, these studies do not address hate content or the issue of content moderation. 

To the extent that content moderation is carried out with artificial intelligence (AI) algorithms, our research is also related to the growing stream of literature on the implications of applying AI algorithms. A number of papers have studied the application in fintech \citep{wei2015credit}, hiring \citep{cowgill2019economics, lee2018understanding}, online dating \citep{abeliuk2019price}, and advertising \citep{lambrecht2019algorithmic}. Our study differs from  these papers in that we theoretically explore the strategic implications of using AI algorithms in content moderation and also the incentives for platforms to perfect their technology.

Content moderation is a hotly debated issue in political science, communications, and economics, and many of the discussions involve free speech, censorship, and the merits or demerits of content moderation \citep[e.g.,][]{gillespie2018custodians,myers2018censored,gorwa2020algorithmic}. In a complementary theoretical paper, \cite{madio2020user} study content moderation as a tool to attract content-sensitive advertisers and as a way to manage its advertising price. The content moderation we study focuses on the interactions between a platform and users and it differs from theirs in three ways. First, the user content subject to content moderation in our model is the one that all users do not like, to a varying degree. In their case, it is the content that some users like and some do not, and overall users' demand for the platform actually goes up with more such content. Second, content moderation in our paper is motivated by a platform's effort to please and attract users, while their content moderation is motivated solely by attracting advertisers. Third, we examine how a platform's content moderation strategy interacts with technology and whether a platform under different revenue models has sufficient incentive to perfect its technology, and they do not.

The rest of the paper is structured as follows.
In Section \ref{sec: model},
we develop our theoretical model and discuss a platform's content moderation strategies with perfect technology. In Section \ref{sec:imperfect AI}, we discuss how imperfect technology can affect content moderation and what incentives a platform faces in developing a better technology. Section \ref{sec:policy} explores the policy implications of our model.  Finally, in Section~\ref{sec: conclusion}, we conclude.

\section{Model} \label{sec: model}

Consider a social media platform, with users of mass of 1, where they post their opinions and read those from others. Users are heterogeneous with respect to how extreme their expressed opinions typically are. We use the extremeness index $x\in[0,1]$ to capture this heterogeneity. Thus, a user located at $x$ tends to express opinions with extremeness index $x$. We assume that users are distributed uniformly over the index range, i.e.,  $x\sim U[0,1]$. When a user posts content, as literature has shown \citep[e.g.,][]{bazarova2014self,buechel2015motivations}, she gains a utility of $u(x)$ from sharing her opinions on the platform. This utility differs amongst users. Literature in consumer psychology shows that individuals with more extreme opinions are also more vocal in expressing their opinions \citep{miller2009expressing,yildirim2013user,mathew2019spread}. Based on this finding, we model the utility from posting content on social media as  $u(x)=\alpha x$, where $\alpha\geq0$, such that a user with a higher extremeness index gains more utility from posting.
Here a larger $\alpha$ implies a greater difference in posting utilities between any two users.

A user on the platform also derives utility from reading content posted by others.  We assume that a user with extremeness index $x$ on the platform gains the maximum reading utility of $v$ if the platform is populated with posts of like-minded users or less extreme ones. In other words, if all posts have  extremeness index $x$ or less, the user at $x$ gains $v$. 
However, a user at $x$
will find objectionable a post with extremeness index $\tilde{x}>x$ and her utility is reduced by $\tilde{x}$ per post with the same index.  In other words, all the more extreme posts will reduce a user's reading utility. Mathematically, the utility for a user at $x$ from reading the posts in the extremeness index range of $[0,\overline{x}]$ where $\overline{x}>x$
is given by $v-\int_{x}^{\overline{x}}\tilde{x}d\tilde{x}$. We assume a user is exposed to all content on the platform. We further assume $v<\frac{1}{2}$ to ensure that  the least extreme users ($x=0$) have a negative utility if she is exposed to all the content on the platform without moderation.

The platform can moderate the extreme user-generated content. Due to the large volume of UGC, it typically relies on artificial intelligence (AI) and natural language processing algorithms to identify and remove intended content. We start in our benchmark model with the assumption of a perfect content moderation technology such that a platform can get rid of any extreme content with perfect accuracy. This means that a platform can choose a $y\in[0,1]$ such that any content with extremeness index $x>y$ is eliminated while any content with $x\leq y$ is kept. Later in our analysis (Section \ref{sec:imperfect AI}), we will also look into imperfect technologies where the platform cannot perfectly moderate the intended content but can only remove any content $x>y$ and preserve any content  $x\leq y$ with a higher than random probability. This imperfect technology nests our benchmark model as a special case.





If the platform engages in content moderation and deletes a user's content because it is deemed offensive, then the user experiences a psychological cost $c$ when she could not post or her post cannot be seen by others. 
This cost captures how people treasure  freedom of expression. Throughout the model, we shall maintain the assumption $c>v$ to ensure that users subject to content moderation with certainty will not participate in the platform. 
Without loss of generality, we set $c\leq\alpha+2v$. This assumption is sufficient to guarantee that a user can still participate in the platform facing uncertain prospect of content moderation, as we will see in Section \ref{sec:imperfect AI}.

The anticipated utility of a user from participating in a social media platform $U(x)$ is the sum of utilities from both reading and posting content. A user participates in the platform if $U(x)\geq0$. Mathematically, $U(x)$ is given by
\begin{equation}\label{eq:utility_fcn}
 U(x)=
    \begin{cases}
        \underbrace{\alpha x}_{\text{posting utility}} + \underbrace{v- \int_{\tilde{x}\in\hat{\mathcal{X}},x<\tilde{x}\leq y}\tilde{x}d\tilde{x}}_{\text{reading utility}} & \text{if $x\leq y$,} \\
        \underbrace{- c}_{\text{posting utility}} + \underbrace{v}_{\text{reading utility}} & \text{if $x>y$.} 
  \end{cases}
\end{equation}
where $\hat{\mathcal{X}}$ is expected set of participants on the platform by a user located at $x$.

Depending on whether a platform uses advertising or subscription as its revenue model, it may earn revenues from advertisers or from users through subscription fees. Specifically, we assume an advertising model such that the platform charges the advertisers for each customer on the platform. This means that the total advertising revenue is proportional to the number of users on the platform, i.e., $$\pi^{A}=\zeta X^A,$$ where $X^A$ is the user base of the platform, 
and $\zeta$ is the advertising value of each user, or ARPU (average revenue per user). 
If the platform earns revenue from subscription instead, it sets a subscription fee $p$ to its users. Denote the number of users who choose to use the platform when it charges $p$ as
$X^S(p)$. 
Then the platform's revenue is $$\pi^S=pX^S(p).$$ 
In our model, the subscription fee is endogenously determined by the platform whereas the per user advertising fee ($\zeta$) is determined by a competitive market and exogenous to our model.\footnote{As a standard practice, advertising fees on major social media platforms like Facebook and Twitter are not set by the platforms but determined through auctions.} 

By juxtaposing these two revenue models, we can examine the incentives a platform faces in content moderation in each of the revenue models. Moreover, we can also examine how the ability of a platform to conduct content moderation may influence the choice of its revenue model. The timeline of the game is as follows:



\begin{enumerate}
    \item If the platform uses advertising, it takes the advertising fee ($\zeta$) as given; if it uses subscription, it sets the subscription fee ($p$) for all users. 
    \item The platform determines its content moderation strategy ($y$).
    \item Users decide whether to stay on the platform and those who stay on post content. Users read the content remaining on the platform after moderation, and obtain utility from posting and reading. 
\end{enumerate}

In the following section, we first solve for the equilibrium for a given revenue model. That is, we analyze the platform's content moderation behavior separately for both advertising-based and subscription-based revenue models. Then, we will discuss how content moderation can affect a platform's preference for revenue models.


\subsection{Advertising-Supported Social Media Platforms}

We start the analysis by considering the case for an ad-supported platform. Let the users who actually participate in the platform be $\mathcal{X}$.
$\hat{\mathcal{X}}$, as introduced above in Equation (\ref{eq:utility_fcn}), is a user's expected set of participants on the platform. Each user will decide whether or not to participate in the social media platform based on her utility $U(x)$ and we derive the equilibrium where $\mathcal{X}=\hat{\mathcal{X}}$ \citep{katz1985network,easley2010networks}.

We start by characterizing the equilibrium configuration for users for any given content moderation policy $y$. The following lemma summarizes our analysis. 

\begin{lemma}\label{lem:config}
For any $y\in[0,1]$, there exists $x^A(y)\in[0,y]$ such that in equilibrium, the set of people who participate in the platform $\mathcal{X}$ is a continuum on $[x^A(y),y]$. Furthermore, $U(x)$ is increasing in $x$ on $[x^A(y),y]$.
\end{lemma}

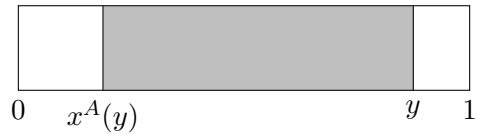
\begin{figure}[h]
\caption{Illustration of the platform's user base (advertising)}
    \label{fig:config_ad}
    \centering
\begin{tikzpicture}[scale=0.75]

\path [fill=lightgray] (7,0) -- (1.5,0) -- (1.5,1.5) -- (7,1.5);

\draw (0,0) node [below] {$0$} -- (8,0) node [below] {$1$} -- (8,1.5) -- (0,1.5)  --cycle;
\draw (7,0) node [below] {$y$} -- (7,1.5);
\draw (1.5,0) node [below] {$x^A(y)$} -- (1.5,1.5);

\end{tikzpicture}
    
\end{figure}

The proof of Lemma \ref{lem:config} is on p. \pageref{pf:lem1} in Appendix \ref{appdx:pfs}. This lemma is illustrated in Figure \ref{fig:config_ad} and the platform's user base is given by the shaded area. Based on Lemma 1, we know that in equilibrium, for any $x\in[x^A(y),y]$, the utility of a user with extremeness index $x$ from participating on the platform can be expressed as: 
\begin{equation}\label{eq:utility_fcn_eqm}
    U(x)=\alpha x+v-\int_{x}^{y}\tilde{x}d\tilde{x}.
\end{equation}
Based on the utility function, we can solve for $x^A(y)$ by setting $U(x^A(y))$ to zero, for any given level of content moderation $y$. The solution is given below:
\begin{equation}\label{eq:x_lower_bar}
    x^A(y)=
    \begin{cases}
        -\alpha+\sqrt{\alpha^2+y^2-2v} & \text{if $y\geq\sqrt{2v}$,} \\
        0 & \text{if $y<\sqrt{2v}$.} 
  \end{cases}
\end{equation}
Here, when $y\geq\sqrt{2v}$, we have a regime of a more lax content moderation. In this case, more content moderation will decrease $x^A(y)$, or draw more users of less extreme views to the platform. This is the case where moderating extreme views may help the platform to expand its market. When $y<\sqrt{2v}$, we have a regime of a more strict content moderation. In this case, $x^A(y)$ is bounded at zero and  more content moderation will simply reduce the platform's customer base. 

Thus, the revenue of a platform under advertising is given by
\begin{equation}
    \pi^A=\zeta(y-x^A(y)),
\end{equation}
and the platform chooses the optimal level of content moderation, $y^{A*}$, to maximize its revenue $\pi^A$. The following proposition summarizes the optimal content moderation strategy of a platform under advertising.

\begin{proposition}\label{prop:ad_eq}
{\bf \em (Advertising model and content moderation)} 
A platform with advertising as its revenue model does not always have incentives to conduct content moderation. It will conduct content moderation $y^{A*}=\sqrt{2v}$ only if the  posting utility in the market is sufficiently small relative to the maximum reading utility, or $\alpha<{\alpha^A}\equiv\sqrt{2v}$. The optimal revenue is given by $\pi^{A*}=\zeta\sqrt{2v}$. Otherwise, the platform does not moderate content ($y^{A*}=1$) and its optimal revenue is given by $\pi^{A*}=\zeta(1+\alpha-\sqrt{\alpha^2+1-2v})$.

\end{proposition}

The proof of Proposition \ref{prop:ad_eq} is on p. \pageref{pf:prop1} in Appendix \ref{appdx:pfs}.
Proposition \ref{prop:ad_eq} suggests that content moderation is a tool for the platform to achieve its revenue objectives. When a platform uses advertising as its revenue model, it needs to maximize its user base to maximize its advertising revenues. Content moderation can help the platform to maximize its user base if by cutting extreme content and pruning extreme users expands the user base amongst the less extreme users. We can see this more clearly by comparing the marginal gain and marginal loss in market size associated with content moderation. 

Note that at any given $y$, the platform will always want to do more content moderation $dy$ if doing so can draw more users $dx^A$ to the platform, or algebraically $dx^A>dy$. This implies that if $\frac{d x^A}{d y}>1$, the platform will do content moderation to expand its customer base. It is simple to show that $\frac{d x^A}{d y}=\frac{y}{\sqrt{y^2+\alpha^2-2v}}$ and it is larger than 1 if and only if $\alpha<\sqrt{2v}\equiv\alpha^A$, which is the condition given in Proposition \ref{prop:ad_eq}. 

To probe deeper, content moderation can expand the platform's customer base fundamentally because moderating extreme views and hence pruning extreme users on a platform  can attract more users with less extreme opinions by increasing their reading utility. To see this, if a platform wants to expand its customer base beyond $x^A$ to the left by $\Delta x$, the new users have a lower posting utility by $\alpha \Delta x$. These new users will only participate in the platform if their reading utility is increased by $\alpha \Delta x$. The platform can increase the reading utility for the marginal user only by further content moderation by the amount of $\Delta y$. The amount of reading utility increase by $\Delta y$ is given by the expression $\Delta y\frac{\alpha y}{\sqrt{y^2+\alpha^2-2v}}$. Thus, the amount of $\Delta y$ needed to increase the marginal user's reading utility by $\alpha\Delta x$, denoted as $\Delta\tilde{ y}$, is given by $\Delta \tilde{y}=\frac{\Delta x\sqrt{y^2+\alpha^2-2v}}{y}$. 
This expression decreases with a higher $v$, or at a higher $v$ a smaller change in content moderation is required to deliver the same amount of reading utility to the marginal users. This explains why at the optimal content moderation $y^{A*}=\sqrt{2v}$, the platform does less content moderation when $v$ is large. This also explains why a platform would not do any content moderation if $v$ is too small (the threshold $\alpha^A\equiv\sqrt{2v}$ is too small): too much content moderation to deliver too little reading utility. 
As a platform's objective is to maximize its customer base, it will do more content moderation as long as $\Delta \tilde{y}<\Delta x$, which also gives us the condition in Proposition \ref{prop:ad_eq}.

Through this analysis we can see that Proposition \ref{prop:ad_eq} reveals an interesting insight about content moderation across different platforms. Under the advertising revenue model, whenever users care sufficiently more about reading utility than posting utility, the platform is motivated to moderate content. Otherwise, it is not. This insight seems consistent with our casual observations. Users on the social media platform Parler, for instance, seem to care more about posting and we see little content moderation there. Whereas users on YouTube seem to care more about viewing than posting, we see more strict content moderation.

\subsection{Subscription-Supported Social Media Platforms}\label{sec:subsc_model}

When the revenue source is subscription fees, the platform determines a content moderation strategy ($y$) as in the case of advertising model, and also sets a subscription fee ($p$) for all users. Then, a user at $x$ will participate in the platform if her net utility $U(x)-p>0$. Similar to the advertising case, a platform's customer base is illustrated in Figure \ref{fig:config_subsc}.
\begin{figure}[t]
\caption{Illustration of the platform's user base (subscription)}
    \label{fig:config_subsc}
    \centering
\begin{tikzpicture}[scale=0.75]

\path [fill=lightgray] (7,0) -- (2,0) -- (2,1.5) -- (7,1.5);

\draw (0,0) node [below] {$0$} -- (8,0) node [below] {$1$} -- (8,1.5) -- (0,1.5)  --cycle;
\draw (7,0) node [below] {$y$} -- (7,1.5);
\draw (2,0) node [below] {$x^S(y,p)$} -- (2,1.5);

\end{tikzpicture}
    
\end{figure}
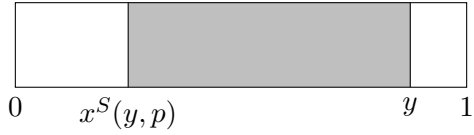
The marginal user
$x^S(y,p)$, which is dependent also on $p$ now, is given by
\begin{equation}\label{eq:xS}
 x^S(y,p)=
    \begin{cases}
        -\alpha+\sqrt{\alpha^2+y^2-2(v-p)} & \text{if $y\geq\sqrt{2(v-p)}$,} \\
        0 & \text{if $y<\sqrt{2(v-p)}$.} 
  \end{cases}  
\end{equation}
The platform maximizes its subscription revenue $\pi^S$ by setting its content moderation strategy $y$, and  subscription fee $p$, and the revenue is given by
\begin{equation}
    \pi^S=p(y-x^S(y,p)).
\end{equation}
The following proposition summarizes the optimal strategy for the platform under a subscription revenue model.

 \begin{proposition}\label{prop:subsc_eq}
 {\bf \em (Subscription model and content moderation)} 
Under the subscription model, there exists $\alpha^S\in(0,\alpha^A)$, such that
the platform will conduct content moderation if $\alpha<\alpha^S$. 
The optimal content moderation strategy, equilibrium subscription fee, and the resulting revenue are given respectively by $y^{S*}=\sqrt{\frac{2v}{3}}$, $p^{*}={\frac{2v}{3}}$, and $\pi^{S*}=(\frac{2v}{3})^{\frac{3}{2}}$.
Otherwise, if $\alpha\geq\alpha^S$, the platform does not moderate content ($y^{S*}=1$). The platform's optimal price and profit are given by $p^*=p_1^*\equiv\frac{1}{9}\Big[(1+\alpha)\sqrt{2(2-3v+2\alpha^2+\alpha)}-2(1-3v+\alpha^2-\alpha)\Big]$, and  $\pi^{S*}=p_1^*(1+\alpha-\sqrt{\alpha^2+1-2(v-p_1^*)})$.
\end{proposition}

The proof of Proposition \ref{prop:subsc_eq} is on pp. \pageref{pf:prop2}-\pageref{pf:prop3} in Appendix \ref{appdx:pfs}.
By analyzing Proposition \ref{prop:subsc_eq}, we can develop insights about what motivates a platform under subscription to do more or less content moderation, how content moderation affects its pricing, and finally how content differs under advertising vs subscription revenue models because of a platform's effort in content moderation.

By comparing Propositions \ref{prop:ad_eq} and \ref{prop:subsc_eq}, we see that 
a platform under advertising is more likely to do content moderation ($\alpha^A>\alpha^S$). This seems to be consistent with casual observations. Facebook and Twitter are two prominent examples of advertising-based platforms, and they both actively moderate content. Even though a platform under subscription is less likely to moderate its content, when it does, it will moderate content more aggressively than what it would  under advertising ($y^{S*}=\sqrt{\frac{2v}{3}}<y^{A*}=\sqrt{2v}$). The platform is less likely to engage in content moderation because the subscription fee screens out less extreme users on the platform, and the remaining users are more extreme, getting less disutility from other more extreme users. For that reason, moderating extreme content adds less utility to the marginal users under subscription than to those under advertising.  In other words, content moderation is less effective in attracting marginal users when subscription model is used or $\frac{\partial x^S}{\partial y}<\frac{\partial x^A}{\partial y}$, all else being equal. This explains why content moderation is more sparingly used under subscription.

The reason why a platform under subscription may behave more aggressively once it decides to do content moderation is related to the role of pricing. Under subscription, a platform can use pricing to internalize its decision on the extent of content moderation, which is not possible under advertising. To see this clearly, we can derive how the optimal price for the platform may change with its content moderation decision, and the expression in a general form is given by
$$\frac{\partial p^*}{\partial y}=\frac{1-\frac{\partial x^S}{\partial y}-p^*\frac{\partial^2 x^S}{\partial p\partial y}}{2\frac{\partial x^S}{\partial p}+p^*\frac{\partial^2 x^S}{\partial y^2}}.$$
The denominator of this expression is positive guaranteed by the second-order condition. Therefore, content moderation will lead to a lower price by the platform if the numerator is positive, which is the case if content moderation adds little utility to the marginal users (a small $\frac{\partial x^S}{\partial y}$), or content moderation increases price sensitivity on the part of marginal users (a large $|\frac{\partial^2 x^S}{\partial p\partial y}|$). Thus, price can help the platform to expand its user base more effectively in conjunction with content moderation so that it wants to do it more aggressively. If the numerator is negative, which is the case if marginal users are very responsive to content moderation but their price sensitivity does not change much with content moderation, the platform will increase its price with content moderation. This is the case where the platform internalizes content moderation efforts by charging a higher subscription fee. Given our modeling assumptions, price is used to enhance the user expansion effect of content moderation.\footnote{In our model, users' disutility from reading extreme content comes from all users with higher extremeness indices. This assumption, although more realistic, reduces the response of marginal users to a platform's content moderation. If we were to let a user's disutility only come from the most extreme content on the platform, we would enhance this response greatly so that the platform will want to raise its price to internalize any content moderation.} 

The conclusions that a platform under advertising is more likely to conduct content moderation and that when conducting content moderation, a platform under subscription does more aggressively can both be tested with suitable data.
To provide some  prima facie evidence, we have collected data on 103 social media platforms based on the ``101+ Social Media Sites You Need to Know in 2021''  composed by \textit{Influencer Marketing Hub}.\footnote{\url{https://influencermarketinghub.com/social-media-sites/}.} As shown in Appendix \ref{sec:empirical}, we collect the texts of their content moderation policy and also  information about their revenue models. In addition, we hire independent graders from Mechanical Turk to read and code the texts of content moderation policy for each platform.
Our analysis shows that out of all the social media platforms in our analysis, only two  platforms do not conduct content moderation and they both adopt subscription as their major revenue model. Our regression analysis further shows that the platforms with advertising as their revenue model tend to have a less restrictive content moderation policy than those with subscription (see Appendix \ref{sec:empirical} for details). 
While not being conclusive, these findings are consistent with the conclusions coming out of our theoretical analysis, providing some preliminary external validity for our modeling efforts.

The first two propositions also allow us to shed light on whether content tends to be more or less extreme on a platform with subscription vs advertising model as a result of conducting content moderation. Our analysis shows that whenever a platform under subscription does not moderate content, it has more extreme content and appeals to more extreme users than a platform under advertising. However, when a platform under subscription does conduct content moderation, it fields less extreme content and caters to less extreme users than a platform under advertising. This is because subscription fee serves to screen out less extreme users when a platform does not moderate content, and when it does, as discussed previously, it uses content moderation more aggressively and charges a lower subscription fee to draw moderate users to the platform. This analysis offers a testable hypothesis that platforms under subscription tend to have the most extreme or the least extreme content. Anecdotal evidence seems consistent with this hypothesis. Gab, for instance, is a subscription-based platform with extreme content \citep{zannettou2018gab}.  LinkedIn is a subscription-based platform that does strict content moderation to ensure that the social discourse on the platform remains professional.


\subsection{Content Moderation and Revenue Models}

The previous two sections show that a platform's revenue model will influence its content moderation strategy. In this section, we push that line of inquiry one step further to see how content moderation can influence the choice of a platform's revenue model. We will do so by first examining when a platform may choose subscription over advertising, and then comparing that choice with the choice when content moderation is not allowed. 

A platform will choose advertising over subscription if and only if $\pi^{A*}$ is larger (smaller) than $\pi^{S*}$. This comparison will define a $\overline{\zeta}$ such that a platform will choose advertising if and only if $\zeta>\overline{\zeta}$. Here $\overline{\zeta}$ is the minimum advertising value per user needed for a platform to embrace advertising model. Thus, a larger $\overline{\zeta}$ will make it less likely for a platform to choose advertising.
Similarly, when content moderation is not allowed, we can define a $\hat{\zeta}$ such that the platform will choose advertising if and only if $\zeta>\hat{\zeta}$. By comparing  $\overline{\zeta}$  and $\hat{\zeta}$, we can isolate how optimal content moderation can alter a platform's preference for advertising vs subscription model. We will make the comparison for all $\alpha<\alpha^A$. For any $\alpha\geq\alpha^A$, we have the trivial case where content moderation makes no difference in the choice of revenue model because no content moderation will be conducted regardless of this choice, even if content moderation is allowed. 
The following proposition summarizes the findings.

\medskip 
\begin{proposition}\label{prop:biz model choice}
{\bf \em (Content moderation and revenue model choice)} 
Relative to the case of no content moderation, a platform conducting optimal content moderation is more likely to choose subscription over advertising ($\overline{\zeta}>\hat{\zeta}$) if the maximum posting utility is sufficiently small, i.e., $\alpha<\alpha_1$. Otherwise, i.e., $\alpha_1<\alpha<\alpha^A$, optimal content moderation increases the likelihood of a platform choosing advertising  ($\overline{\zeta}<\hat{\zeta}$).

\end{proposition}

\begin{figure}[h]
\caption{Revenue model choice and content moderation}\label{fig:aBar_alpha}
\centering
\includegraphics[width=0.7\textwidth]{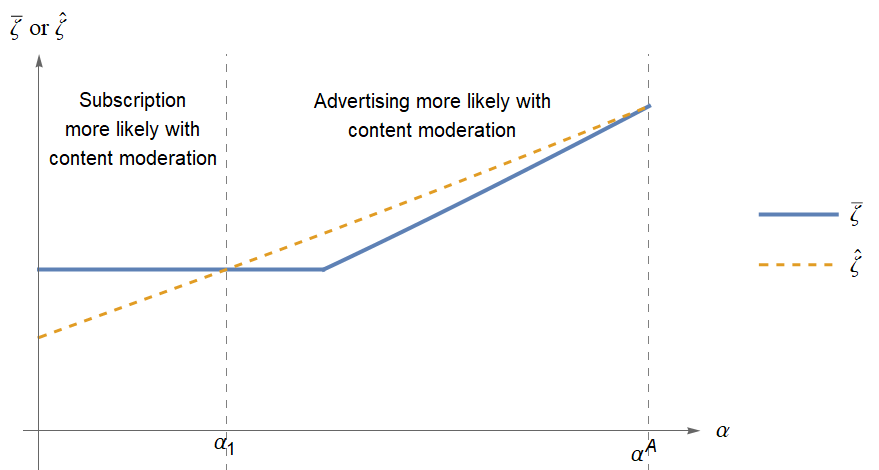}
\end{figure}

The proof of Proposition \ref{prop:biz model choice} is on pp. \pageref{pf:prop3}-\pageref{pf:prop4} in Appendix \ref{appdx:pfs}. Proposition \ref{prop:biz model choice} is illustrated in Figure \ref{fig:aBar_alpha}. For $\alpha<\alpha_1$, we see $\overline{\zeta}>\hat{\zeta}$ in Figure \ref{fig:aBar_alpha}, implying that it takes a higher advertising rate per user for a platform to choose advertising over subscription model when content moderation is introduced. For $\alpha>\alpha_1$, we have $\overline{\zeta}<\hat{\zeta}$, which implies that a platform is willing to embrace advertising model at a lower advertising rate when content moderation is allowed. This proposition suggests an intriguing insight that at a low $\alpha$, it requires advertisers to pay a higher advertising rate to switch a platform from subscription to advertising model when content moderation is allowed. Equivalently, at a low $\alpha$, content moderation makes it more likely for a platform to adopt the subscription model at a given advertising rate.

Intuitively, content moderation helps a platform under subscription more than that under advertising because without content moderation, marginal users are more sensitive to any change in maximum posting utility ($\alpha$) under subscription than advertising, and hence a platform under subscription suffers more in profitability with any reduction in $\alpha$.
Content moderation neutralizes that effect to deliver more profit gain to a platform under subscription.

Proposition \ref{prop:biz model choice} suggests a testable hypothesis that in the environment where social media platforms are free to conduct content moderation vs one where platforms are constrained for one reason or another, we shall see more platforms choosing to adopt a subscription model over an advertising model if users care more about reading than posting. Otherwise, we would expect the social media to use advertising model more. The variation in the extent of content moderation across Europe, US, and China may provide a good testing ground for this hypothesis.

\section{Content Moderation and Technology}\label{sec:imperfect AI}

The analysis in the previous section delivers the key insights into the incentives a platform  faces under advertising or subscription in using  content moderation and also the impact of content moderation on the choice of revenue model. These insights are delivered under the assumption of perfect technology for content moderation. In this section, we shall expand our analysis and explore a platform's content moderation strategy under imperfect technology.

In reality, an accurate technology for content moderation is still many years into the future \citep{AIprovesisapoor}. As  Mark Zuckerberg commented, ``over a five-to-ten-year period we will have AI tools that can get into some of the linguistic nuances of different types of content to be more accurate, to be flagging things to our systems, but today we're just not there on that... There's a higher error rate than I'm happy with'' \citep{timelineforAIto}. In a well-publicized example, in the days leading up to 4th of July, 2018,  Facebook's algorithm for ``hate speech detection'' flagged down and removed a post of the Declaration of Independence because of paragraphs 27-31, which include the phrase ``merciless Indian savages'' \citep{FacebookBI}. 
The existence of imperfect technology  raises a number of questions about the practice and management of content moderation.

First, if technology has a ``higher error rate'' than a platform is ``happy with,'' how should a platform employ the technology given the choice of its revenue model? In this regard, a related question is whether a platform has the incentive to embrace an inaccurate technology to do content moderation? Second, how can a platform best manage its content moderation to achieve its profit objectives? Given that a platform's primary objective is to maximize its profit, could content moderation with imperfect technology lead to a higher extremeness index for the platform? Finally, if today's technology is ``just not there,'' and there is ``a higher error rate,'' what kind of a platform has the most incentives to improve it or not improve it? In this section, we address all those questions by extending our model to incorporate imperfect technology for content moderation.

When a platform uses imperfect technology, it can err in two ways. On the one hand, it may not be able to prune the extreme content a platform wants to eliminate completely so that part of the extreme content remains on the platform. On the other hand, it may accidentally prune the content it wants to preserve. To capture both types of errors and also to nest our main model as a special case, we specify the content moderation technology $q_k(x|y)$ as the probability that a content generated by a user with extremeness index $x$ is removed by the platform when it intends to prune all $x>y$ given its technology accuracy $k$. Specifically, we have:
\begin{equation}\label{eq:qk}
    q_k(x|y) = 
    \begin{cases}
    \frac{1}{2}-k & \text{if $x
    \leq y$;} \\
    \frac{1}{2}+k & \text{if $x>y$.} 
  \end{cases}
\end{equation}

Above technology prunes any content $x>y$ with probability $\frac{1}{2}+k$, where $k\in[0,\frac{1}{2}]$. It also accidentally deletes any content $x<y$ with probability $\frac{1}{2}-k$. In other words, the technology allows a platform to prune extreme content with a higher probability than it deletes moderate content accidentally. When $k=\frac{1}{2}$, we go back to our main model where extreme content is cut with perfect accuracy. When $k=0$, all content on the platform is cut with equal probability and we have a random technology at work. Thus, a higher $k$ indicates a more accurate technology. Except for this imperfect technology, we also focus our analysis on $\alpha<{\alpha^S}$\label{alpha<alphaS} such that at $k=\frac{1}{2}$ a platform always chooses to do content moderation regardless of whether it is under advertising or subscription models. Then, with this assumption, whenever a platform does not want to do content moderation, it will be due to imperfect technology. We maintain all other assumptions in the previous section. Our analysis will unfold by first looking at content moderation in advertising, then in subscription, and finally the incentives a platform faces in advancing its content moderation technology.

\subsection{Content Moderation with Imperfect Technology}

Due to imperfect technology, when a platform tries to prune all content $x>y$, the content by users with extremeness index $x$ will be eliminated with probability $q_k(x|y)$  and it remains on the platform with probability $1-q_k(x|y)$, as defined in equation (\ref{eq:qk}). Therefore, a user's expected utility from posting is given by $\alpha x\big(1-q_k(x|y)\big)-c q_k(x|y)$ and her utility from reading is correspondingly adjusted by the probability. We can write the total utility  for a user at ${x}$ as
\begin{equation}\label{eq:utility_fcn_2}
    U({x}) = \underbrace{\alpha {x}\big(1-q_k({x}|y)\big) - c q_k({x}|y)}_{\text{posting utility}} + \underbrace{v- \int_{\tilde{x}\in\hat{\mathcal{X}},\tilde{x}>{x}}\tilde{x}\big(1-q_k(\tilde{x}|y)\big)d\tilde{x}}_{\text{reading utility}},
\end{equation}
which is a generalization of equation (\ref{eq:utility_fcn}).

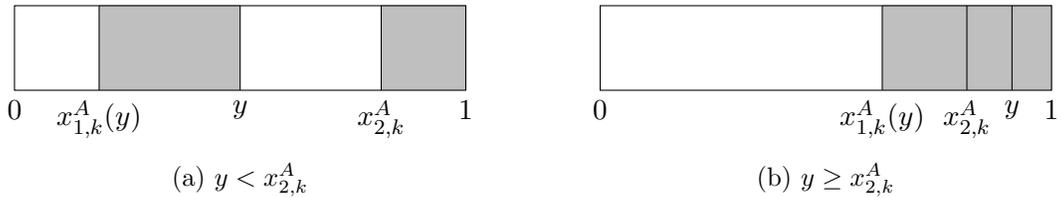
\begin{figure}[h]
\caption{User base of an ad-supported platform with  imperfect content moderation technology}\label{fig:config_k_ad}
\centering
  \begin{subfigure}[b]{0.45\columnwidth}
  \centering
    \begin{tikzpicture}[scale=0.75]
\path [fill=lightgray] (8,0) -- (6.5,0) -- (6.5,1.5) -- (8,1.5);
\path [fill=lightgray] (4,0) -- (1.5,0) -- (1.5,1.5) -- (4,1.5);

\draw (0,0) node [below] {$0$} -- (8,0) node [below] {$1$} -- (8,1.5) -- (0,1.5)  --cycle;
\draw (4,0) node [below] {$y$} -- (4,1.5);
\draw (1.5,0) node [below] {$x^A_{1,k}(y)$} -- (1.5,1.5);
\draw (6.5,0) node [below] {$x^A_{2,k}$} -- (6.5,1.5);
\end{tikzpicture}
    \caption{$y<x_{2,k}^A$}
    \label{fig:config_k_ad:1}
  \end{subfigure}
  \begin{subfigure}[b]{0.45\columnwidth}
  \centering
   \begin{tikzpicture}[scale=0.75]
\path [fill=lightgray] (8,0) -- (5,0) -- (5,1.5) -- (8,1.5);

\draw (0,0) node [below] {$0$} -- (8,0) node [below] {$1$} -- (8,1.5) -- (0,1.5)  --cycle;
\draw (7.3,0) node [below] {$y$} -- (7.3,1.5);
\draw (5,0) node [below] {$x^A_{1,k}(y)$} -- (5,1.5);
\draw (6.5,0) node [below] {$x^A_{2,k}$} -- (6.5,1.5);
\end{tikzpicture}
    \caption{$y\geq x_{2,k}^A$}
    \label{fig:config_k_ad:2}
  \end{subfigure}
\end{figure}

As we show in Appendix \ref{appdx:impeft_ai_ad}, whenever a platform conducts content moderation, the users on the platform fall into one of the two configurations illustrated in Figure \ref{fig:config_k_ad}. In Figure \ref{fig:config_k_ad:1}, content moderation creates two disjoint segments of users and in Figure \ref{fig:config_k_ad:2}, we have a contiguous user segment, all dependent on the extent of content moderation $y$. In this figure, the variables $x_{2,k}^A$ and $x_{1,k}^A(y)$ are respectively given as:
\begin{equation}
x^A_{2,k} = 
    \begin{cases}
    \sqrt{\alpha^2+1+\frac{2c(1+2k)-4v}{1-2k}}-\alpha<1 & \text{if $k\leq\overline{k}$;} \\
    1 & \text{if $k>\overline{k}$,} 
  \end{cases}
\end{equation}
where $\overline{k}=\frac{\alpha+2v-c}{2(\alpha+c)}$, and 
\begin{equation}
 x^A_{1,k}(y) = 
    \begin{cases}
    \sqrt{\alpha^2+\max\big\{0,y^2+\min\{2\alpha y,\frac{(1-2k)(2c+1-(x^A_{2,k})^2)-4v}{1+2k}\}\big\}}-\alpha & \text{if $y<x^A_{2,k}$;} \\
    \sqrt{\alpha^2+\max\big\{0,y^2+\min\{2\alpha y,\frac{(1-2k)(2c+1-y^2)-4v}{1+2k}\}\big\}}-\alpha & \text{if $y\geq x^A_{2,k}$.}
  \end{cases}
\end{equation}
Furthermore, $x_{2,k}^A$ is increasing in $k$ and $x_{1,k}^A(y)$ is decreasing in $k$.

Interestingly,  if the technology is not sufficiently good ($k<\overline{k}$), then $x^A_{2,k}<1$, which means  the most extreme users in $[x^A_{2,k},1]$ will stay on the platform regardless of how the platform conducts its content moderation, as extreme users always derive the highest utility from the platform.
In other words, inaccurate technology can no longer screen out the most extreme users, a fact about imperfect technology that a platform under advertising can benefit from, as we will see soon.

Similar to advertising case, when a platform adopts subscription, we can show in Appendix \ref{appdx:impeft_ai_subsc} that we have similarly well-defined user configurations as in Figure \ref{fig:config_k_ad}. As in Section \ref{sec:subsc_model}, the platform once again chooses its moderation strategy and subscription price. We refer readers to Appendix \ref{appdx:impeft_ai_subsc} for detailed analysis of this case. 

With imperfect technology, it is no longer the case that a platform can remove extreme content with perfect accuracy. For that reason, a platform's strategy in content moderation will always involve pruning the extreme content it intends to (all content at $x>y$) and also the moderate content it does not (all content at $x\leq y$), and depending on the technology it uses (for a given $k$), it may prune the extreme content more than the moderate content and vice versa.
Then, two related questions arise. First, to maximize its profit, should a platform always prune the extreme content more than the moderate content? Second, from a user's perspective, when a platform prunes the moderate content more than the extreme content, does the platform always have a high average extremeness index\footnote{The average extremeness index on the platform can be calculated as $\overline{x}=\frac{\int_{\mathcal{X}}x(1-q(x))dx}{\int_{\mathcal{X}}(1-q(x))dx}$\label{foot:average_ext},
where the numerator is the weighted sum of the location index (weighted by the probability of not being removed), which represents the ``total'' extremeness  of all remaining content, and the denominator is the expected number of posts remaining after content moderation.}? The following proposition summarizes a platform's optimal content moderation strategy.

\begin{proposition}\label{prop:cm w imperfect ai}
{\bf \em (Content moderation with imperfect technology)} 
For both advertising and subscription, a platform will conduct content moderation only if technology is sufficiently accurate. When conducting content moderation, the platform may prune the moderate content more than the extreme content, but the average extremeness index on the platform may be lower than when it prunes the extreme content more or does not prune any content.
\end{proposition}

The proof of Proposition \ref{prop:cm w imperfect ai} can be found on pp. \pageref{pf:prop4}-\pageref{pf:prop5} in Appendix \ref{appdx:pfs}. Proposition \ref{prop:cm w imperfect ai} first suggests that a platform needs a sufficiently accurate technology to start  content moderation. Secondly, when technology is sufficiently good,
the optimal content moderation strategy may call for a platform
to prune the moderate content more than the extreme content. This is because a sufficiently accurate technology already deters extreme users from participating the platform, so there is little extreme content on the platform in the first place.
Finally, whether or not a platform prunes extreme content more than the moderate content is not a good yardstick to judge whether a platform is extreme or moderate. This means that a user looking to join a platform may not find a moderate outlet even if the outlet is pruning a lot of extreme content. This may be because there are many extreme users on the platform in the first place. 

However, the case where pruning the moderate content may lead to a low extremeness index deserves a closer look. The optimal content moderation strategy calls for pruning the moderate content more than the extreme content when there is little extreme content on the platform in the first place. Then the question is, why prune moderate content if there is little extreme content on the platform? The reason is strategic. With a blunt instrument, or imperfect technology, pruning the moderate content is the collateral damage to pruning the extreme content, or the price a platform pays to reduce extreme content. Thus, it may be necessary for a platform to prune only the moderate content in order to deter extreme users from ever getting onto the platform. 

Proposition \ref{prop:cm w imperfect ai} suggests three managerial as well as policy insights about content moderation. First, no one should be alarmed about a platform pruning moderate content or not eliminating extreme content, and it is part of a platform's optimal strategy when technology is imperfect. For this reason, we may see more social media executives blaming technology. Second, as technology improves, a platform's optimal strategy will become increasingly controversial and increasingly harder to justify. This is because when technology is sufficiently accurate, the optimal strategy calls for the platform to prune the moderate content more than the extreme content.  Finally, the content moderation strategy by a platform and the diligence with which it is pruning the extreme content  may not tell the full story about how extreme the content may be on the platform. To tell the full story, one will have to also consider the technology used and the preferences of the user base.

\subsection{Content Moderation  and Incentive for Technology Improvement}

As technology improves, a platform's strategy in content moderation will also change. In this regard, our model of imperfect technology allows us to shed light on two related questions. First, will a platform impose a more strict standard for content moderation when technology improves? Second, as the content moderation strategy of a platform also affects its profitability, does a platform actually have an incentive to improve the technology? The following two propositions suggest some nuanced answers to these two questions.

\begin{proposition}\label{prop:k_and_ystar}
{\bf \em (Better technology and less content moderation)} 
When technology is  sufficiently accurate (a sufficiently large $k$), a platform under either advertising or subscription will adopt a more relaxed standard for content moderation as technology further improves. As a result, the average extremeness index increases.
\end{proposition} 

The proof of Proposition \ref{prop:k_and_ystar} is on p. \pageref{pf:prop5} in Appendix \ref{appdx:pfs}. Intuitively, as technology improves, the platform can prune extreme content more accurately to keep moderate marginal users happy so that it does not need to prune as much. In addition, by pruning the extreme content less, the platform can increase its customer base to increase its profit when it is under advertising, and keep more of its high willingness-to-pay users on the platform when it is under subscription.

\begin{proposition}\label{prop:choice_of_k}
{\bf \em (Incentive for imperfect content moderation technology)} 
Under advertising, the platform may choose imperfect technology even if there is no cost involved in improving the technology when the cost to users subject to pruning ($c$) is small. Under subscription, the platform always chooses a perfect technology.
\end{proposition} 

\begin{figure}[h]
\caption{Platform profit and technology accuracy ($v=0.25, \alpha=0.2, c=0.3$)}\label{fig:pi_k}
\centering
  \begin{subfigure}[b]{0.45\columnwidth}
  \centering
   \includegraphics[width=\textwidth]{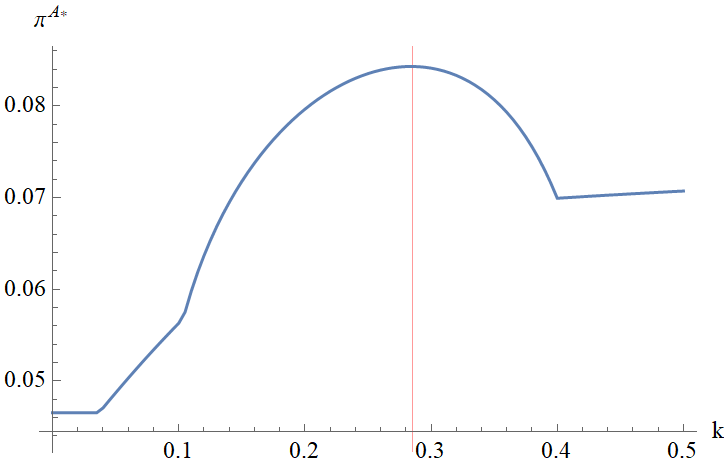}
    \caption{Advertising ($\zeta=0.1$)}
    \label{fig:pi_k_ad}
  \end{subfigure}
  \begin{subfigure}[b]{0.45\columnwidth}
  \centering
   \includegraphics[width=\textwidth]{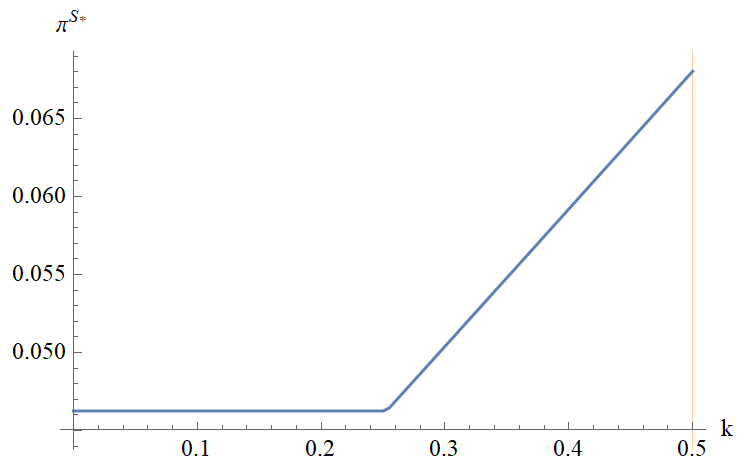}
    \caption{Subscription}
    \label{fig:pi_k_subsc}
  \end{subfigure}
\end{figure}

The proof of Proposition \ref{prop:choice_of_k} is on p. \pageref{pf:prop6} in Appendix \ref{appdx:pfs}, and the proposition is illustrated in Figure \ref{fig:pi_k}. A platform under advertising may not want to develop a perfect technology because its primary objective is to maximize its  customer base. When $c$ is small, the imperfect technology can benefit the platform because a less accurate technology benefits extreme users more than the loss it imposes on moderate users when technology is very accurate in the first place. Therefore, a less accurate technology will increase the number of extreme users more than it reduces the number of moderate users, thus increasing the installed customer base for the platform. However, too much extreme content on the platform will alienate moderate users. This effect becomes increasingly dominant when technology is at lower accuracy (smaller $k$). This explains why in Figure \ref{fig:pi_k_ad} we have an inverted U-shaped relationship between accuracy and platform profit under advertising. When $c$ is sufficiently large, the segment of moderate users is also sufficiently large relative to the segment of extreme users because the cost carries more weight for the extreme users as they have a higher probability of being pruned. In this case, a less accurate technology can still increase the segment of extreme users but it will impose unintended damage on the relatively large segment of moderate users. Therefore, the most effective way to increase the installed customer base is not to increase extreme users but to retain and expand moderate users by reducing the likelihood of unintended pruning, which is to increase the accuracy ($k$). This is why when $c$ is sufficiently large, even a platform under advertising will be motivated to pursue the perfect technology.

In the case of subscription, however,
maximizing a platform's customer base can no longer maximize the platform's profit as the platform has the subscription fee as the second instrument. With this second instrument, the platform can fully internalize the benefit of technology improvement. This shows in the fact that both the platform's customer base and optimal fee increase with technology improvement. Therefore, costs aside, the platform always has the incentive to pursue the perfect technology. 

Propositions \ref{prop:k_and_ystar} and \ref{prop:choice_of_k} offer two rather surprising perspectives on content moderation and technology. First, as content moderation technology becomes more accurate, one should not expect that a profit-maximizing platform will always do more content moderation and publish more moderate content. Second, a number of executives of large  social media platforms, including those of Facebook and Twitter \citep{SocialMediagiantswarnofAI,timelineforAIto}, often complain about the limits of technology in content detection. However, our analysis also suggests an intriguing possibility that a platform under advertising may not have the incentive to pursue a more accurate technology in the first place.   

These two perspectives suggest to a manager that conducting content moderation is not as simple as just reducing the extremeness index. In the pursuit of a better technology, the optimal strategy calls for a manager to relax the criteria for pruning and to increase the average index. Moreover,  the cost to users subject to pruning is an important parameter to watch. When that cost is low, imperfect technology is conducive to attracting a large installed customer base to the platform. When it is large, technology improvement is always a winning strategy. 

One testable hypothesis from Propositions \ref{prop:k_and_ystar} and \ref{prop:choice_of_k} is that when conducting content moderation, all else being equal, we will expect to observe that platforms under subscription have a better technology for content moderation than for those under advertising.



\section{Content Moderation and Policy Implications}\label{sec:policy}

Content moderation is a hotly debated issue that has many policy implications. Many questions are raised in this context. For instance, do platforms have sufficient incentives to conduct content moderation on their own relative to what is optimal for the society? When they do conduct moderation, are they doing too much or too little? Is the technology that is optimal for platforms also optimal for the society?
We can address all these questions with our model by investigating how a social planner will conduct content moderation to maximize social welfare. Our answers can then inform the ongoing debate on whether and how much the government  should get involved in regulating online content and how the regulatory effort may need to be nuanced with regard to platforms of different revenue models.  

To conduct our analysis, we note that the objective of a social planner in content moderation is to maximize social welfare, which is the sum of the utilities for all users for the platform, and the expression of the social welfare, denoted as $W(y)$, when technology is perfect is given by
\begin{equation}\label{eq:welfare}
    W(y)=\int_{x^{P}(y)}^y \left(\alpha  x+v-\int_{x}^y\tilde{x}d\tilde{x}\right)  dx,
\end{equation}
where $x^P(y)$ is the marginal user who is indifferent between participating in the platform and not. As we show 
in the proof of Proposition \ref{prop:social_alpha_y}, the social planner will conduct content moderation if and only if $\alpha<{\alpha^P}$, where at $\alpha^P$ the social planner is indifferent between conducting content moderation and not. When the social planner does conduct content moderation ($\alpha<{\alpha^P}$), the optimal content moderation strategy is given by $y^{P*}=\frac{1}{2}(\alpha+\sqrt{\alpha^2+4v})$. By comparing what the social planner does with what a platform under advertising or subscription does, we have the following proposition.

\begin{proposition}\label{prop:social_alpha_y}
{\bf \em (Social planner's content moderation strategy)}
All else being equal, a social planner is less likely to conduct content moderation than a platform under either advertising or subscription ($\alpha^P<\alpha^S<\alpha^A$). When it does, it adopts a more relaxed standard for content moderation than a platform under subscription, but a more strict one than a platform under advertising ($y^{S*}<y^{P*}<y^{A*}$).
\end{proposition}

The proof of Proposition \ref{prop:social_alpha_y} is on p. \pageref{pf:prop7} in  Appendix \ref{appdx:pfs}. This proposition suggests three insights about the content moderation strategy in  a decentralized market. First, left to market forces, a platform with the profit motivation has even more incentives to engage in content moderation than a social planner. This is because content moderation is less effective at increasing social welfare than at increasing a platform's profitability. Second, more incentives for a platform do not mean right incentives. To maximize the welfare, the social planner only prunes users with a negative utility contribution to the society.\label{check_u_contrib}\footnote{This is checked in Appendix on page \pageref{app:check_u_contrib}.} A user's utility contribution to the society includes her posting utility, reading utility, and the total negative utility her post imposes on other more moderate users. A platform under advertising will keep users with negative utility contribution all for the purpose of maximizing its installed customer base. A platform under subscription will prune users even with positive utility contribution to increase the willingness-to-pay of other users. Third, because of the different incentives that platforms with different revenue models have in content moderation and also because of the severity with which different platforms are motivated to prune content for their own profitability, any regulatory measures may need to account for the difference in platforms' revenue models. In other words, sweeping regulations for all social media platforms regardless of their revenue models could be ill-advised.

Indeed, revenue models also provide different incentives for a platform to perfect its technology. A natural question arises: what technology a social planner would prefer for content moderation, perfect or imperfect? The following proposition addresses this question.

\begin{proposition}\label{prop:social_optimal_k}
{\bf \em (Social planner's technology preference) }
When a social planner conducts content moderation, it always prefers a better technology (higher accuracy $k$) such that, cost aside, it always pursues the perfect technology ($k=\frac{1}{2}$).
\end{proposition}

The proof of Proposition \ref{prop:social_optimal_k} is  on p. \pageref{pf:prop8} in  Appendix \ref{appdx:pfs}. Comparing Proposition \ref{prop:social_optimal_k} with Proposition \ref{prop:choice_of_k}, we see that a platform under advertising does not always have the right incentives to perfect its technology, unless $c$ is sufficiently large. However, a platform under subscription does have the right incentive to develop the technology for content moderation, although the technology is applied in a way that is not socially optimal, as discussed in Proposition \ref{prop:social_alpha_y}.

\section{Conclusion} \label{sec: conclusion}

Content moderation on social media platforms is an important issue that has attracted increasing attention in the past few years from practitioners, scholars, social activists, policy makers, and regulators alike. At a high level, the issue concerns the freedom of expression, political discourse, personal liberty, civil society, and government regulations. At a more basic level, it is a platform's marketing decisions, like any other product or service company would do, on what revenue models to use, what content to allow or what ``product'' to design, and what kind of users to attract or to discourage, all for the purpose of achieving its highest revenues. In addressing this complex issue, it is quite understandable that experts with different objectives offer different perspectives as to whether platforms should do self regulation themselves, or a government intervention is needed to regulate social media content.  In this paper, we take a first step to unpack this complex issue and investigate how a self-interested social media platform may conduct content moderation, how its content moderation strategy may hinge on its revenue model and technology, and what incentives a platform with advertising or subscription as its primary revenue model may have in perfecting content moderation technology. This investigation not only offers normative insights about how a self-interested platform will or will not do content moderation, but also shed light on whether government interventions are needed and if they are, what those interventions may entail.

Our analysis shows that a self-interested platform does not need to care about any social cause to actively engage in content moderation. It can use 
content moderation as a tool to perform two marketing functions: to expand its user base and to increase the willingness-to-pay of the users on its platform.  These dual functions are rooted in the nature of social media  where users gain utilities from posting and reading user-generated content on a platform, but they are also sensitive to content more extreme than what they prefer. For a social planner who cares about social welfare, content moderation is a tool to eliminate users who make negative utility contributions to society.  In this regard, we show that self-interested platforms are more likely to use the dual functions and conduct content moderation than a social planner. In other words, platforms are more eager than a social planner to conduct content moderation motivated by their own self-interest.

Because a self-interested platform conducts content moderation for profit, the economics dictate that its strategy will depend on its revenue model and hence the resulting content on the platform, as measured by the extremeness index, will also depend on the same. We show that in the absence of any content moderation, all else being equal, a platform under subscription revenues will field more extreme content than a platform under advertising. However, when content moderation is conducted, a platform under subscription revenues will curate a more moderate content than  one under advertising. Casual observations seem to support this theoretical insight. Interestingly, the social planner will conduct content moderation to achieve a body of content that is more extreme than under subscription, but more moderate than under advertising.

For most social platforms, technology for content moderation is imperfect as many executives have readily admitted \citep{SocialMediagiantswarnofAI,timelineforAIto}. Our analysis shows that a platform's strategy in content moderation critically depends on the technology it uses. A platform may choose not to do any content moderation at all if its technology is not sufficiently accurate. When it is, a platform may conduct content moderation in an unexpected way. Under imperfect technology, a platform may  throw away the moderate content more than the extreme content as part of its optimal strategy. We show that when this happens, it does not necessarily result in a more extreme platform. Conversely, when a platform prunes the extreme content more than the moderate content, we do not necessarily have a more moderate platform. In other words, one cannot judge how extreme a platform is by looking at its content moderation strategy. This insight is especially germane to policy makers when they try to reduce hate content on a platform by focusing on the removal of hate content upon user complaints, such as what is currently practiced in EU \citep{EUhatespeechreg}.

It is common for social media executives to blame imperfect technology for some lapses in content moderation and those blames are well-placed as our analysis shows. However, our analysis also sheds some light on whether a self-interested platform actually has incentives to perfect its content moderation technology. A platform under advertising may not pursue the perfect technology, even if doing so is costless. We further show that a platform under subscription will pursue the perfect technology as does a social planner. 

Overall, our analysis shows that self-interested platforms are motivated to do content moderation, but their strategy diverges from a social planner's. In this sense, there can be grounds for government interventions. We show that such interventions can only be effective if they are differentiated and nuanced  according to revenue models and   technology levels that different platforms are adopting. 

As managerial insights, our analysis has articulated the marketing roles that content moderation plays in achieving a platform's profit objectives. It also prescribes the normative strategies that a platform can use in content moderation: what content to moderate for what purpose and what strategic adjustments to make with regard to revenue models and technology. Finally, platforms under subscription are well advised to invest in their technology for content moderation. 

Content moderation as a research topic is a target-rich area. We hope our research kindles some interest in this important and timely subject. Future research can take a number of directions. First, in our model, we do not pursue the consequence of unintended content moderation by a platform. One can imagine that there might exist activist users on the platform who may react to unjustifiable pruning. These user reactions may affect a platform's content moderation strategy as well as its incentive to develop a better technology. Second, in our model, we identify the difference between what a social planner will do with content moderation and what a self-interested platform will do, thus probing into the rationale for and approach toward any regulatory interventions. Future research can develop concrete regulatory measures that can induce platforms under advertising or subscription to conduct content moderation in alignment with a social planner. Lastly, many of our theoretical insights seem consistent with anecdotal evidence. Future research can empirically put them to a test.

\bigskip
\setstretch{1.2}
\bibliographystyle{apalike}
\bibliography{lit}

\newpage
\setstretch{1.4}
\newpage 
\setcounter{page}{1}
\singlespacing

\setcounter{section}{0}
\setcounter{subsection}{0}
\setcounter{figure}{0}
\setcounter{table}{0}
\setcounter{equation}{0}
\setcounter{proposition}{0}
\setcounter{lemma}{0}

\renewcommand{\thesection}{A.\arabic{section}}
\renewcommand{\thesubsection}{A.\arabic{section}.\arabic{subsection}}
\renewcommand{\thefigure}{A\arabic{figure}}
\renewcommand{\thetable}{A\arabic{table}}
\renewcommand{\theequation}{A\arabic{equation}}
\renewcommand{\thepage}{A\arabic{page}}
\renewcommand{\theproposition}{A.\arabic{proposition}}
\renewcommand{\thelemma}{A.\arabic{lemma}}

\section*{ONLINE APPENDIX}

\section{Proofs of Lemmas and Propositions in Main Text}\label{appdx:pfs}

\begin{proof}[\bf Proof of Lemma \ref{lem:config}]\label{pf:lem1}

When $x>y$, since $c>v$ by assumption, $U(x)=-c+v<0$ holds and users with $x>y$ do not  participate in the platform.

For users $x \leq y$, consider two users $x_1, x_2$ such that $x_1<x_2\leq y$. Then 
\begin{align*}
    U(x_1)-U(x_2)= & \alpha x_1-\int_{\tilde{x}\in\hat{\mathcal{X}},x_1<\tilde{x}\leq y}\tilde{x}d\tilde{x}-(\alpha x_2-\int_{\tilde{x}\in\hat{\mathcal{X}},x_2<\tilde{x}\leq y}\tilde{x}d\tilde{x})\\
    = &\alpha (x_1-x_2)-\int_{\tilde{x}\in\hat{\mathcal{X}},x_1<\tilde{x}\leq x_2}\tilde{x}d\tilde{x}\\
    < & 0,
\end{align*}
implying first that, if $x_1$ participates then all users in the range $[x_1,y]$ participate, and second that the utility of participating users is increasing in $x$. The former also implies that there exists $x^A\geq 0$ such that the user base is $[x^A,y]$.
\end{proof}

\medskip

\begin{proof}[\bf Proof of Proposition \ref{prop:ad_eq}]\label{pf:prop1}

Recall that the platform's profit maximization problem is $\max_{y}\pi^A=\zeta(y-x^A(y))$, where 
\begin{equation*}
    x^A(y)=
    \begin{cases}
        -\alpha+\sqrt{\alpha^2+y^2-2v} & \text{if $y\geq\sqrt{2v}$,} \\
        0 & \text{if $y<\sqrt{2v}$,} 
  \end{cases}
\end{equation*}
from Equation (\ref{eq:x_lower_bar}).

When $y<\sqrt{2v}$, $x^A(y)=0$ and the profit of the platform is $ \zeta y$, which is maximized when $y^*=\sqrt{2v}$.

When $y\geq\sqrt{2v}$, the profit becomes $ \zeta (y + \alpha- \sqrt{\alpha^2+y^2-2v})$ and $\frac{d\pi^A}{dy}=\zeta(1-\frac{y}{\sqrt{y^2+\alpha^2-2v}})$. Notice that the profit is increasing in $y$ when $\frac{d\pi^A}{dy}>0$, which holds iff  $\alpha>\sqrt{2v}$. Therefore, when $\alpha\geq\sqrt{2v}$, profit maximizing  moderation strategy is $y^{A*}=1$.  
On the other hand, when $\alpha<\sqrt{2v}$,  $\frac{d\pi^A}{dy} \leq 0$ and the profit maximizing content moderation strategy is $y^{A*}=\sqrt{2v}$. In other words, we find $\alpha^A\equiv\sqrt{2v}$ such that  $y^{A*}=\sqrt{2v}$ when $\alpha<\alpha^A$ while $y^{A*}=1$ when $\alpha\geq\alpha^A$. 

\end{proof}

\medskip
\begin{proof}[\bf Proof of Proposition \ref{prop:subsc_eq}]\label{pf:prop2}
Under subscription, the profit maximization problem of 
the platform is $\max_{y,p}\pi^S=p(y-x^S(y,p))$, where 
\begin{equation*}
 x^S(y,p)=
    \begin{cases}
        -\alpha+\sqrt{\alpha^2+y^2-2(v-p)} & \text{if $y\geq\sqrt{2(v-p)}$,} \\
        0 & \text{if $y<\sqrt{2(v-p)}$.} 
  \end{cases}  
\end{equation*}
from Equation (\ref{eq:xS}). 

\begin{itemize}
    \item If $y\leq \sqrt{2(v-p)}$, then $p\leq v-y^2/2$ holds. Thus
    \begin{equation}\label{eq:a1}
       \pi^S= py \leq (v-\frac{y^2}{2})y  \leq \frac{2v}{3}\sqrt{\frac{2v}{3}}.
    \end{equation}
   
Therefore, $\pi^S$ takes the maximum value of  $\frac{2v}{3}\sqrt{\frac{2v}{3}}$, when both inequalities in Equation (\ref{eq:a1}) are equality, i.e., $p=v-\frac{y^2}{2}$ and $(v-\frac{y^2}{2})y=\frac{2v}{3}\sqrt{\frac{2v}{3}}$ These two conditions give the optimal price  $p=\frac{2v}{3}$ and content moderation policy
$y=\sqrt{\frac{2v}{3}}$ when $y\leq \sqrt{2(v-p)}$.

\item If $y\geq \sqrt{2(v-p)}$, then $\pi^S=p(y+\alpha-\sqrt{\alpha^2+y^2-2(v-p)})$. First fix $p$ as given. Taking the first 
order partial derivative w.r.t. $y$, we have 
$\frac{\partial \pi^S}{\partial y}=p\Big(1-\frac{y}{\sqrt{y^2+\alpha^2-2(v-p)}}\Big)$.
If $\alpha^2-2(v-p)>0$, i.e., $p>v-\alpha^2/2$, 
$\frac{\partial \pi^S}{\partial y}=p\Big(1-\frac{y}{\sqrt{y^2+\alpha^2-2(v-p)}}\Big)>p\Big(1-\frac{y}{\sqrt{y^2+0}}\Big)=0$. Therefore, $\pi^S$ is increasing in $y$ for any given $p>v-\alpha^2/2$, or the optimal level of $y$ is 1.
If $\alpha^2-2(v-p)\leq 0$, i.e., $p\leq v-\alpha^2/2$, 
$\frac{\partial \pi^S}{\partial y} \leq0$, which means that $\pi^S$ is decreasing in $y$, or the optimal $y$ for any given $p\leq v-\alpha^2/2$ is $\sqrt{2(v-p)}$.

\end{itemize}

So the optimal level of moderation is either $y=\sqrt{2(v-p)}$ or $y=1$. 
When $y=\sqrt{2(v-p)}$, we have seen that the optimal level of $p$ and $y$ should be $p=\frac{2v}{3}$ and $y=\sqrt{\frac{2v}{3}}$, which induces $\pi^{S}=\frac{2v}{3}\sqrt{\frac{2v}{3}}$. 
When $y=1$, the optimal subscription fee $\hat{p}^{S*}$ should maximize $$\hat{\pi^S}(p)=p(1-x^S(1,p))=p(1+\alpha-\sqrt{\alpha^2+1-2(v-p)}).$$ Solving the first order condition (FOC) w.r.t. $p$ gives $$\hat{p}^{S*}=\frac{1}{9}\Big[(1+\alpha)\sqrt{2(2-3v+2\alpha^2+\alpha)}-2(1-3v+\alpha^2-\alpha)\Big].$$ The second order condition (SOC) is clearly satisfied since $$\frac{\partial^2 \hat{\pi^S}(p)}{\partial p^2}=\frac{-2 (\alpha ^2+1)-3 p+4 v}{(\alpha ^2+2 p-2 v+1)^{3/2}}<\frac{-2+4v}{(\alpha ^2+2 p-2 v+1)^{3/2}}< 0.$$

Therefore, $\pi^{S*}=\max\{\frac{2v}{3}\sqrt{\frac{2v}{3}},\hat{\pi^S}(\hat{p}^{S*})\}$. By the envelope theorem, we know that $$\frac{\partial \hat{\pi^S}(\hat{p}^{S*})}{\partial \alpha}=\frac{\partial \hat{\pi^S}({p})}{\partial \alpha}|_{p=\hat{p}^{S*}}=p(1-\frac{\alpha}{\sqrt{\alpha^2+1-2(v-p)}})|_{p=\hat{p}^{S*}}=\hat{p}^{S*}(1-\frac{\alpha}{\sqrt{\alpha^2+1-2(v-\hat{p}^{S*})}}).$$ Since by construction $1\geq x^S(1,p)=-\alpha+\sqrt{\alpha^2+1-2(v-p)}$, we know that $1-\frac{\alpha}{\sqrt{\alpha^2+1-2(v-p)}}\geq0$ for any $p$, and specifically for $p=\hat{p}^{S*}$. Thus, $\frac{\partial \hat{\pi^S}(\hat{p}^{S*})}{\partial \alpha}\geq0$, i.e., $\hat{\pi^S}(\hat{p}^{S*})$ is increasing in $\alpha$. $\frac{2v}{3}\sqrt{\frac{2v}{3}}$ is independent of $\alpha$. Therefore, proving that there exists $\alpha^S\in(0,\alpha^A)$ such that $\pi^{S*}=\frac{2v}{3}\sqrt{\frac{2v}{3}}$ (with the optimal content moderation strategy $y^{S*}=\sqrt{\frac{2v}{3}}$) when $\alpha<\alpha^S$ and $\pi^{S*}=\hat{\pi^S}(\hat{p}^{S*})$ (with the optimal content moderation strategy $y^{S*}=1$) when $\alpha>\alpha^S$ requires $\hat{\pi^S}(\hat{p}^{S*})<\frac{2v}{3}\sqrt{\frac{2v}{3}}$ when $\alpha=0$ and $\hat{\pi^S}(\hat{p}^{S*})>\frac{2v}{3}\sqrt{\frac{2v}{3}}$ when $\alpha=\sqrt{2v}\equiv\alpha^A$ for any $v\in(0,\frac{1}{2})$. To check this, we denote $$H(v)=(\hat{\pi^S}(\hat{p}^{S*})-\frac{2v}{3}\sqrt{\frac{2v}{3}})|_{\alpha=0}$$ and $$J(v)=(\hat{\pi^S}(\hat{p}^{S*})-\frac{2v}{3}\sqrt{\frac{2v}{3}})|_{\alpha=\sqrt{2v}}$$ and we want to show $H(v)<0$ and $J(v)>0$ for any $v\in(0,\frac{1}{2})$.

Since $H(v)=(\hat{\pi^S}(\hat{p}^{S*})-\frac{2v}{3}\sqrt{\frac{2v}{3}})|_{\alpha=0}$, plugging in the expression of $\hat{\pi^S}(\hat{p}^{S*})$ and $\alpha=0$ obtains $$H(v)=\frac{1}{27} (2- \sqrt{4-6 v}) (6 v+ \sqrt{4-6 v}-2)-\frac{2 v}{3}\sqrt{\frac{2 v}{3}},$$ $$H'(v)=\frac{1}{3}(2-\sqrt{6v} -\sqrt{4-6 v})$$ and $$H''(v)=\frac{1}{\sqrt{4-6 v}}-\frac{1}{\sqrt{6v}}.$$ Note that $H''(v)\lessgtr0$ if $v\lessgtr\frac{1}{3}$, so $H'(v)$ is first decreasing and then increasing on $(0,\frac{1}{2})$. Therefore, $H'(v)<\max\{H'(0),H'(\frac{1}{2})\}=\max\{0,\frac{1}{3}(1-\sqrt{3})\}=0$, so $H(v)$ is decreasing in $v$. Thus, $$H(v)<H(0)=0.$$ 

Since $J(v)=(\hat{\pi^S}(\hat{p}^{S*})-\frac{2v}{3}\sqrt{\frac{2v}{3}})|_{\alpha=\sqrt{2v}}$, plugging in the expression of $\hat{\pi^S}(\hat{p}^{S*})$ and $\alpha=\sqrt{2v}$ obtains $$J(v)=\frac{1}{27} (2 \sqrt{2 v}-\sqrt{2 (v+\sqrt{2 v}+2)}+2) (2 v+2 \sqrt{2 v}+(2 \sqrt{v}+\sqrt{2}) \sqrt{v+\sqrt{2 v}+2}-2),$$ $$J'(v)=\frac{(\sqrt{2}-3 \sqrt{6}) v+\sqrt{v+\sqrt{2 v}+2}+\sqrt{v} (\sqrt{2 (v+\sqrt{2 v}+2)}+2)-\sqrt{2}}{9 \sqrt{v}}.$$ Denote the numerator of $J'(v)$ as $J_1(v)$, then $$J_1'(v)=\frac{4 \sqrt{2} v+8 \sqrt{v}+4 \sqrt{v+\sqrt{2 v}+2}+5 \sqrt{2}}{4 \sqrt{v} \sqrt{v+\sqrt{2 v}+2}}+\sqrt{2}-3 \sqrt{6}.$$ $J_1'(v)=0$ has a unique solution $v=v_0$ where $v_0<\frac{1}{2}$.
Furthermore, $J_1'(v)\gtrless0$ when $v\lessgtr v_0$. Therefore, $J_1(v)>\min\{J_1(0),J_1(\frac{1}{2})\}=\min\{0,\frac{\sqrt{2}-3 \sqrt{6}+2 \sqrt{14}}{2}\}=0$. So, $J'(v)>0$ and then $$J(v)>J(0)=0.$$

Thus, by the fact that $\hat{\pi^S}(\hat{p}^{S*})$ is increasing in $\alpha$ while $\frac{2v}{3}\sqrt{\frac{2v}{3}}$ is independent of $\alpha$, we claim that there exists $\alpha^S\in(0,\alpha^A)$ such that $\pi^{S*}=\frac{2v}{3}\sqrt{\frac{2v}{3}}$ ($y^{S*}=\sqrt{\frac{2v}{3}}$) when $\alpha<\alpha^S$, while $\pi^{S*}=\hat{\pi^S}(\hat{p}^{S*})$ ($y^{S*}=1$) when $\alpha>\alpha^S$.

\end{proof}

\begin{proof}[\bf Proof of Proposition \ref{prop:biz model choice}]\label{pf:prop3}


We first consider the case where content moderation is allowed.
When $\alpha<\alpha^S$, $\pi^{A*}=\zeta\sqrt{2v}$ and $\pi^{S*}=\frac{2v}{3}\sqrt{\frac{2v}{3}}$. Thus, $\overline{\zeta}=\frac{\frac{2v}{3}\sqrt{\frac{2v}{3}}}{\sqrt{2v}}=\frac{2v}{3\sqrt{3}}$ which is independent of $\alpha$. When $\alpha^S\leq\alpha<\alpha^A$, $\pi^{A*}=\zeta\sqrt{2v}$ and $\pi^{S*}=\hat{\pi^S}(\hat{p}^{S*})$. Thus, $\overline{\zeta}=\frac{\hat{\pi^S}(\hat{p}^{S*})}{\sqrt{2v}}$. We have proved that $\hat{\pi^S}(\hat{p}^{S*})$ is increasing in $\alpha$ in the proof of Proposition \ref{prop:subsc_eq}, so $\overline{\zeta}$ is also increasing in $\alpha$.

If content moderation is not allowed, the expressions for the platform's profits (denoted as $\pi^A_0$ and $\pi^S_0$) are the same as those when no moderation is conducted, i.e.,
$\pi^A_0=\zeta(1+\alpha-\sqrt{\alpha^2+1-2v})$ and $\pi^S_0=\hat{\pi^S}(\hat{p}^{S*})$. Thus, $\hat{\zeta}=\frac{\hat{\pi^S}(\hat{p}^{S*})}{1+\alpha-\sqrt{\alpha^2+1-2v}}$. Clearly, when $\alpha^S\leq\alpha<\alpha^A$, $\overline{\zeta}<\hat{\zeta}$ since $\pi^S_0=\pi^{S*}$ but $\pi^A_0<\pi^{A*}$ (this is when the optimal strategy for an advertising-based platform is to  conduct moderation but that of a subscription-based one is not to do so).

We calculate
\begin{align*}
\frac{\partial\hat{\zeta}}{\partial \alpha}= & \frac{\frac{\partial \hat{\pi^S}(\hat{p}^{S*})}{\partial \alpha}(1+\alpha-\sqrt{\alpha^2+1-2v})-\hat{\pi^S}(\hat{p}^{S*})(1-\frac{\alpha}{\sqrt{\alpha^2+1-2v}})}{(1+\alpha-\sqrt{\alpha^2+1-2v})^2}\\
    = & \frac{\hat{p}^{S*}}{(1+\alpha-\sqrt{\alpha^2+1-2v})^2}\Big((1-\frac{\alpha}{\sqrt{\alpha^2+1-2(v-\hat{p}^{S*})}})(1+\alpha-\sqrt{\alpha^2+1-2v})\\
    & -(1+\alpha-\sqrt{\alpha^2+1-2(v-\hat{p}^{S*})})(1-\frac{\alpha}{\sqrt{\alpha^2+1-2v}})\Big).
\end{align*}

Denote $A=\sqrt{\alpha^2+1-2(v-\hat{p}^{S*})}$ and $B=\sqrt{\alpha^2+1-2v}$, then $A\geq B>\alpha$.
\begin{align*}
    \sign(\frac{\partial\hat{\zeta}}{\partial \alpha})= & \sign\Big((1-\frac{\alpha}{A})(1+\alpha-B)-(1+\alpha-A)(1-\frac{\alpha}{B})\Big)\\
    = & \sign\Big(\frac{(A-B)(AB+(1-A-B)\alpha+\alpha^2)}{AB}\Big)\\
    = & \sign(AB+(1-A-B)\alpha+\alpha^2).
\end{align*}
Note that 
\begin{align*}
    AB+(1-A-B)\alpha+\alpha^2)= & A(B-\alpha)-B\alpha+\alpha+\alpha^2\\
    \geq & B(B-\alpha)-B\alpha+\alpha+\alpha^2\\
    = & (B-\alpha)^2+\alpha\\
    >& 0.
\end{align*}
Therefore, $\frac{\partial\hat{\zeta}}{\partial \alpha}>0$, i.e., 
$\hat{\zeta}$ is increasing in $\alpha$.  When $\alpha=\alpha^S$, it has been shown at the end of last paragraph that $\hat{\zeta}>\overline{\zeta}$. When $\alpha=0$,
$$\hat{\zeta}-\overline{\zeta}=\frac{1}{27} \Big(\frac{(3-(\sqrt{4-6 v}+1)) (6 v+\sqrt{4-6 v}-2)}{1-\sqrt{1-2 v}}-6 \sqrt{3} v\Big).$$ Proving  $\hat{\zeta}<\overline{\zeta}$ requires that $$G(v)=(3-(\sqrt{4-6 v}+1)) (6 v+\sqrt{4-6 v}-2)-6 \sqrt{3} v(1-\sqrt{1-2 v})<0$$ for any $v\in(0,\frac{1}{2})$. Note that $$G'(v)=-\frac{3 \sqrt{3}}{\sqrt{1-2 v}}+9 (\sqrt{3-6 v}- \sqrt{4-6 v})-6 \sqrt{3}+18,$$ while both $-\frac{3 \sqrt{3}}{\sqrt{1-2 v}}$ and  $\sqrt{3-6 v}- \sqrt{4-6 v}$ are decreasing in $v$, so $G'(v)$ is decreasing in $v$. Then $G'(v)<G'(0)=0$ so $G(v)$ is decreasing in $v$. Thus, $G(v)<G(0)=0$. Therefore, 
$\hat{\zeta}<\overline{\zeta}$
when 
$\alpha=0$. Thus, by the fact that $\hat{\zeta}$ is increasing in $\alpha$, we can claim that there is $\alpha_1\in(0,\alpha^S)$ such that $\hat{\zeta}\lessgtr\overline{\zeta}$ when $\alpha \lessgtr \alpha_1$ and finish the proof. 

\end{proof}

\begin{proof}[\bf Proof of Proposition \ref{prop:cm w imperfect ai}]\label{pf:prop4}

\noindent {\bf Part (i):}
We first prove that a platform carries out content moderation only if the technology is sufficiently accurate, under both advertising and subscription revenues. To this end, we show that there exists $\epsilon>0$ such that when $k<\epsilon$, the profit induced by optimal moderation strategy is less than that of no content moderation. 
Notice that when $k=\frac{1}{2}$, both  under advertising and subscription revenue, a platform chooses to moderate content, under the assumption $\alpha<\alpha^S$ stated on page \pageref{alpha<alphaS}. Since the platform's profit $$\pi^A_k=\zeta(1-x_{2,k}^A+y-x_{1,k}^A(y))$$ or $$\pi^S_k=p(1-x_{2,k}^S(p)+y-x_{1,k}^S(y,p))$$ is obviously continuous in $k$, it suffices to show that the platform's profit if moderating content is lower than that if no moderation is conducted, when the technology accuracy is $k=0$. 

We start with the analysis of a platform with advertising revenues.  Note that when $k=0$, all content has the probability $\frac{1}{2}$ of being pruned, regardless of their extremeness index $x$ and the platform's choice of $y$. A user at $x=1$ receives utility $U(1)\equiv\frac{1}{2}\alpha-\frac{1}{2}c+v\geq 0$ by the assumption $c\leq \alpha+2v$. Therefore, the user base for the platform will be $[\underline{x},1]$ where the marginal user $\underline{x}$ is the solution to $U(\underline{x})\equiv\frac{1}{2}\alpha\underline{x}-\frac{1}{2}c+v-\frac{1}{2}\frac{1}{2}(1-\underline{x}^2)=0$. With some algebra, we know that the size of the platform's user base is $$1-\underline{x}=1+\alpha-\sqrt{\alpha^2+1-2(2v-c)}<1+\alpha-\sqrt{\alpha^2+1-2v},$$ which is the user base size if no moderation is conducted. The inequality comes from the fact that $v<c$ and thus $2v-c<v$.
Therefore, the profit with content moderation is also less than that without moderation. 

The proof when the platform earns revenues from subscription is similar to that under advertising revenues. With subscription revenues and lowest accuracy ($k=0$), one can show that for any given subscription price $p$, the user base is smaller when the platform moderates content than when it does not: If $p$ induces $U(1)\leq 0$, there is no user on the platform so  the user base (zero) is  trivially smaller  when the platform moderates content than when it does not. If  $p$ induces $U(1)>0$, the procedure to prove that the user base is smaller when the platform moderates content than when it does not is exactly the same as for advertising case, except that we replace the terms $v$ with $v-p$.

Therefore,  we have proved that a platform will conduct content moderation only if technology is sufficiently accurate, for both advertising and subscription.

\bigskip 
\noindent {\bf Part (ii):} Next, we prove that if the platform is moderating content, it may prune more of the moderate content than it does of the extreme content, and moreover, the average extremeness index of the content on the platform may be lower than when it prunes more of the extreme content and when it does not moderate content. To prove the existence of an equilibrium content moderation strategy where these statements hold, it suffices to give an example. 

First, consider a platform under advertising revenue. Based on Figure \ref{fig:config_k_ad}, when the content moderation policy is $y$, the amount of the extreme content (i.e., $x>y$) that is pruned in equilibrium, denoted as $M_{1,k}(y)$, is $$M_{1,k}(y)=(\frac{1}{2}+k)(1-\max\{y,x_{2,k}\}),$$ and the amount of the moderate content (i.e., $x<y$) that is pruned in equilibrium, denoted as $M_{2,k}(y)$, is $$M_{2,k}(y)=(\frac{1}{2}-k)(y-x_{1,k}(y)).$$
In equilibrium, the platform prunes $M_{1,k}^{A*}\equiv M_{1,k}(y^{A*})$ unit of extreme content as well as $M_{2,k}^{A*}\equiv M_{2,k}(y^{A*})$ unit of moderate content. Based on the expression of the average extremeness index ($\overline{x}$) on page \pageref{foot:average_ext}, in equilibrium,  the average extremeness index ($\overline{x}^{A*}_k$) is 
\begin{align*}
    \overline{x}^{A*}_k=&\frac{\int_{\mathcal{X}}x(1-q(x))dx}{\int_{\mathcal{X}}(1-q(x))dx}\\
    =&\begin{cases}
    \frac{\int_{x_{1,k}(y^{A*})}^{y_k^{A*}}x(\frac{1}{2}+k)dx+\int_{x_{2,k}}^{1}x(\frac{1}{2}-k)dx}{(\frac{1}{2}+k)(y_k^{A*}-x_{1,k}(y^{A*}))+(\frac{1}{2}-k)(1-x_{2,k})}\text{\quad if content moderation is conducted in equilibrium,}\\
    \frac{-\alpha+\sqrt{\alpha^2+1-2v}+1}{2}\text{\quad if content moderation is not conducted in equilibrium.}
    \end{cases}
\end{align*}

Consider $\alpha=0.05, v=0.2, c=0.25$. Plug in the numbers into the expressions for $M_{1,k}^{A*}$, $M_{2,k}^{A*}$, and $\overline{x}^{A*}_k$, we can plot out a figure with $k$ as x-axis while  $M_{1,k}^{A*}$, $M_{2,k}^{A*}$, and $\overline{x}^{A*}_k$ as y-axis, to find out whether there can be cases such that the following two claims hold: 

(1) the platform prunes more of the moderate content than it does of the extreme content, i.e., there exists $k_0\in[0,\frac{1}{2}]$ such that $ M_{1,k_0}^{A*}<M_{2,k_0}^{A*}$, and 

(2) the average extremeness index of the content on the platform is lower than when it prunes more of the extreme content and when it does not moderate content, i.e., there exists $k_1,k_2\in[0,\frac{1}{2}]$ such that $ M_{1,k_1}^{A*}<M_{2,k_1}^{A*}$, $ M_{1,k_2}^{A*}>M_{2,k_2}^{A*}$, but $\overline{x}^{A*}_{k_1}>\overline{x}^{A*}_{k_2}$. 

Figure \ref{fig:k_m1_m2_xbar} illustrates the relationship between $k$ and $M_{1,k}^{A*}$, $M_{2,k}^{A*}$, or $\overline{x}^{A*}_k$.

\begin{figure}[h]
\caption{$M_{1,k}^{A*}$, $M_{2,k}^{A*}$, or $\overline{x}^{A*}_k$ and technology accuracy $k$ ($v=0.2, \alpha=0.05$, $c=0.25$)}\label{fig:k_m1_m2_xbar}
  \centering
   \includegraphics[width=0.75\textwidth]{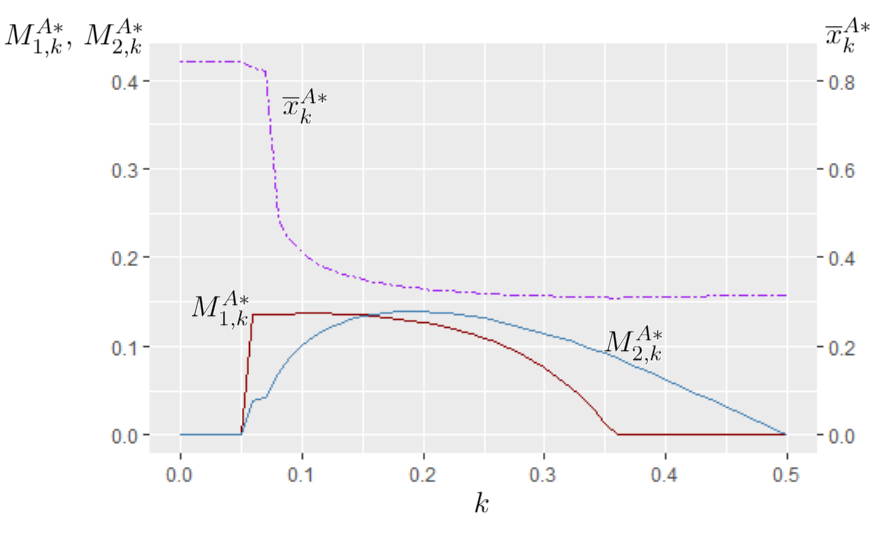}
\end{figure}

From Figure \ref{fig:k_m1_m2_xbar}, we see that for $k$ greater than around 0.15, we have $ M_{1,k_0}^{A*}<M_{2,k_0}^{A*}$, so claim (1) holds. Also, consider $k_1=0.1$ and $k_2=0.3$, we see from the figure that $ M_{1,k_1}^{A*}<M_{2,k_1}^{A*}$, $ M_{1,k_2}^{A*}>M_{2,k_2}^{A*}$, but $\overline{x}^{A*}_{k_1}>\overline{x}^{A*}_{k_2}$, so claim (2) holds.




For a platform under subscription revenues, since the full solution including optimal pricing is analytically challenging (see Section \ref{appdx:impeft_ai_subsc}), we use numerical simulations which exhaust the parameter space of $\alpha\in[0,1]$, $v\in[0,\frac{1}{2}]$, and $c\in[v,\alpha+2v]$ with a grid of 0.05.  The details of how to generate the equilibrium outcomes are in Appendix \ref{appdx:numerical} (especially Section \ref{app:numerical_subsc}). Based on the outcomes stored in Dataframe S (described on page \pageref{tables}), we cannot find any examples where the platform prunes extreme content more than it does moderate content, and we find that the average extremeness index is lower when a platform conducts content moderation than that when it does not.

\end{proof}

\medskip 

\begin{proof}[\bf Proof of Proposition \ref{prop:k_and_ystar}]\label{pf:prop5}

First, consider a platform with advertising revenues.  When $k>\overline{k}$, no users with $x>y$ will participate in the platform. Based on Lemma \ref{lem:cm_k}, the only two candidates for the optimal content moderation $y_k^{A*}$ are either $y_k^{A*}=1$ or $y_k^{A*}=\sqrt{\frac{4v-(1-2k)2c}{1+2k}}$. 

Let the profit of the platform when it chooses $y=1$ and $y=\sqrt{\frac{4v-(1-2k)2c}{1+2k}}$ be $\pi_1$ and $\pi_2$, respectively. 
When $k=\frac{1}{2}$ (i.e., perfect technology), conducting content moderation is more profitable than not doing so, based on the assumption that $\alpha<\alpha^S<\alpha^A$. Therefore, $\pi_1<\pi_2$ when $k=\frac{1}{2}$. Since the platform's profit, $\pi^A_k$ or $\pi^S_k$, is continuous in $k$, there exists a $\hat{k}>\overline{k}$ such that $\pi_1<\pi_2$ for any $k\in[\hat{k},\frac{1}{2}]$. That is, the optimal content moderation strategy is $y_k^{A*}=\sqrt{\frac{4v-(1-2k)2c}{1+2k}}$ for any $k\in[\hat{k},\frac{1}{2}]$. Note that $\frac{\partial{(\frac{4v-(1-2k)2c}{1+2k}})}{\partial k}=\frac{8(c-v)}{(1+2k)^2}>0$, so $y_k^{A*}$ is increasing in $k$, i.e., the platform adopts a more relaxed standard for content moderation as technology further improves. The average extremeness index is $\overline{x}^{A*}_k=\frac{\int_{\mathcal{X}}x(1-q(x))dx}{\int_{\mathcal{X}}(1-q(x))dx}=\frac{\int_0^{y_k^{A*}}x(\frac{1}{2}+k)dx}{(\frac{1}{2}+k)y_k^{A*}}=\frac{y_k^{A*}}{2}$ is increasing in $k$ since $y_k^{A*}$ is increasing in $k$.

The solutions under subscription revenues are proven numerically since characterizing the equilibrium as a closed form solutions is analytically not tractable. We show the monotonicity between $k$ and $y_k^{S*}$ or $\overline{x}$ numerically by exhausting the parameter space of $\alpha\in[0,1]$, $v\in[0,\frac{1}{2}]$, and $c\in[v,\alpha+2v]$ with a grid of 0.05. The details of how to generate the equilibrium outcomes are in Appendix \ref{appdx:numerical} (especially Section \ref{app:numerical_subsc}). Using the outcomes stored in Dataframe S (described on page \pageref{tables}), we can show the relationship between $k$ and the optimal moderation strategy $y_k^{S*}$, as well as the relationship between $k$ and the average extremeness index $\overline{x}^{S*}_k=\frac{\int_{\mathcal{X}}x(1-q(x))dx}{\int_{\mathcal{X}}(1-q(x))dx}=\frac{\int_{x_{1,k}^{S*}}^{y_k^{S*}}x(\frac{1}{2}+k)dx+\int_{x_{2,k}^{S*}}^{1}x(\frac{1}{2}-k)dx}{(\frac{1}{2}+k)(y_k^{S*}-x_{1,k}^{S*})+(\frac{1}{2}-k)(1-x_{2,k}^{S*})}$ numerically.\footnote{When no content moderation is conducted, the average extremeness index $\overline{x}^{S*}_k=\overline{x}_0^S=\frac{x^S(1,p_1^*)+1}{2}$ when the expressions of $x^S(y,p)$ and $p_1^*$ are given by Equation (\ref{eq:xS}) and the last sentence of Proposition \ref{prop:subsc_eq}, respectively.}
Figure \ref{fig:k_yS_k_xbarS} below is an example when $\alpha=0$, $v=0.25$, and $c=0.5$. \\

\begin{figure}[htp]
\caption{Content moderation policy ($y_k^{S*}$) and avg. extremeness index ($\overline{x}^{S*}_k$) vs. technology accuracy ($k$) ($v=0.25, \alpha=0, c=0.5$)}\label{fig:k_yS_k_xbarS}
\centering
  \includegraphics[width=\textwidth]{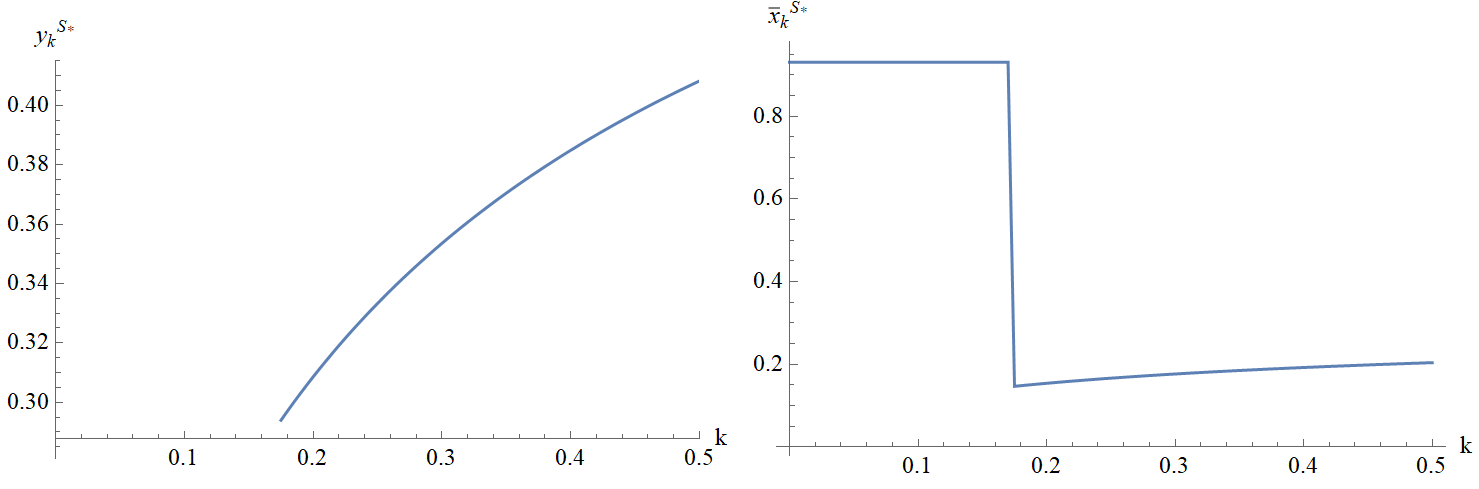}
  \end{figure}
We can see that when $k$ is large, $y_k^{S*}$ and $\overline{x}^{S*}_k$ are increasing in $k$.\footnote{In the left subfigure of Figure \ref{fig:k_yS_k_xbarS}, there is no value of $y_k^{S*}$ for small $k$, which means no content moderation is conducted in equilibrium when $k$ is small.} The same pattern is repeated for all other $(\alpha,v,c)$ combinations.

\end{proof}

\medskip 

\begin{proof}[\bf Proof of Proposition \ref{prop:choice_of_k}]\label{pf:prop6}

For the advertising case, although Section \ref{appdx:impeft_ai_ad} provides the equilibrium for the imperfect technology case,  expressions for the solution are too complicated to analytically derive comparative statics. Therefore, we exhaust the parameter space of $\alpha\in[0,1]$, $v\in[0,\frac{1}{2}]$, and $c\in[v,\alpha+2v]$, using a grid of 0.05, and for any combination of $(\alpha,v,c)$, we find out the maximum profit across different $k$. Details of the numerical solution are provided in Section \ref{app:numerical_ad} of Appendix \ref{appdx:numerical}. 

Then we compare the maximum profit ($\max_k\pi_k^{A*}$) with the profit when $k=\frac{1}{2}$ (i.e., perfect technology, $\pi_{k=\frac{1}{2}}^{A*}$). The numerical results confirm that for any $\alpha$ and $v$, if $c$ is small, we have $\max_k\pi_k^{A*}>\pi_{k=\frac{1}{2}}^{A*}$; otherwise, $\max_k\pi_k^{A*}=\pi_{k=\frac{1}{2}}^{A*}$. Therefore, when $c$ is small, imperfect technology with $k<\frac{1}{2}$ is optimal for a platform under advertising. Figure \ref{fig:pi_k_a} below illustrates the relationship between $k$ and $\pi_k^{A*}$ when $\alpha=0.2$, $v=0.25$. Specifically, Figure \ref{fig:pi_k_smallc} corresponds to the case when $c$ is small while Figure \ref{fig:pi_k_largec} corresponds to the case when $c$ is large. We see that the optimal technology is less than $\frac{1}{2}$ when $c$ is small, but is exactly $\frac{1}{2}$ when $c$ is large. Similar results can be seen for all other  combinations of $(\alpha,v,c)$.

\begin{figure}[h]
\caption{Platform profit $\pi_k^{A*}$ and technology accuracy $k$ ($v=0.25, \alpha=0.2$)}\label{fig:pi_k_a}
\centering
  \begin{subfigure}[b]{0.45\columnwidth}
  \centering
   \includegraphics[width=\textwidth]{figures/pi_k_ad.PNG}
    \caption{$c=0.3$ (small)}
    \label{fig:pi_k_smallc}
  \end{subfigure}
  \begin{subfigure}[b]{0.45\columnwidth}
  \centering
   \includegraphics[width=\textwidth]{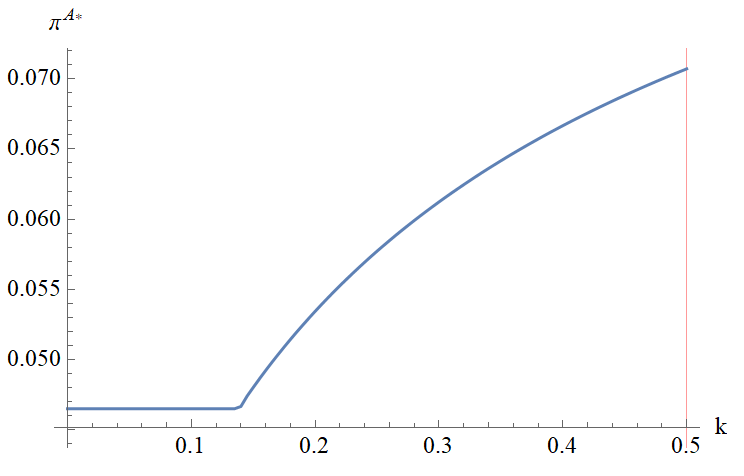}
    \caption{$c=0.5$ (large)}
    \label{fig:pi_k_largec}
  \end{subfigure}
\end{figure}

For the subscription case,  since the full solution including optimal pricing is analytically challenging (see Section \ref{appdx:impeft_ai_subsc}),  we show for any $k\in[0,\frac{1}{2}]$, the optimal profit $\pi^{S*}_{k}$ is weakly increasing in $k$
numerically by exhaustive simulation. Details are provided in Appendix \ref{appdx:numerical}, especially Section \ref{app:numerical_subsc}.
Figure \ref{fig:pi_k_s} below illustrates the relationship between $k$ and $\pi_k^{S*}$ when $\alpha=0.2$, $v=0.25$, and $c=0.3$. A similar pattern can be seen for all other  combinations of $(\alpha,v,c)$.

\begin{figure}[h]
\caption{Platform profit $\pi_k^{S*}$ and technology accuracy $k$ ($v=0.25, \alpha=0.2$, $c=0.3$)}\label{fig:pi_k_s}
  \centering
   \includegraphics[width=0.5\textwidth]{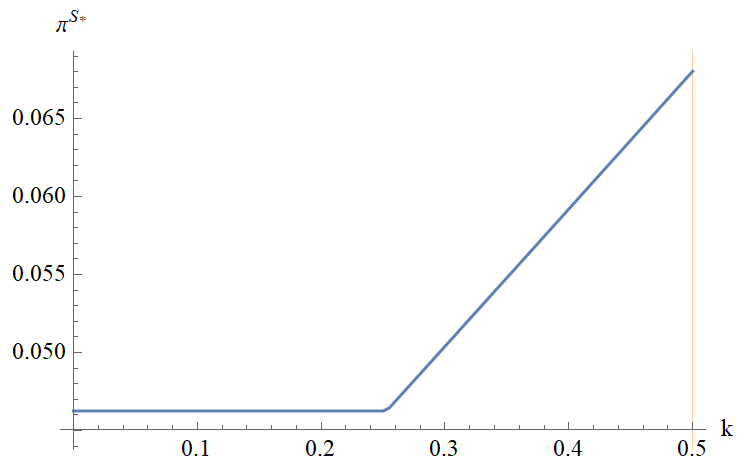}
\end{figure}

\end{proof}

\begin{proof}[\bf Proof of Proposition \ref{prop:social_alpha_y}]\label{pf:prop7}

First, notice that   $x^P(y)=x^A(y)$ since the users' behavior is the same as the advertising case.
Moreover, based on Equations (\ref{eq:x_lower_bar}) and (\ref{eq:welfare}), the social welfare $W(y)$ is given by
\begin{equation*}
    W(y)=
    \begin{cases}
    \int_0^y\left(\alpha  x+v-\frac{1}{2}(y^2-x^2)\right)  dx &\text{ if } y<\sqrt{2v},\\
    \int_{-\alpha+\sqrt{y^2+\alpha^2-2v}}^y\left(\alpha  x+v-\frac{1}{2}(y^2-x^2)\right)  dx &\text{ if } y\geq\sqrt{2v}.
\end{cases}
\end{equation*}

If $y<\sqrt{2v}$, then the FOC with respect to $y$ is $\frac{dW(y)}{dy}=\alpha y+v-y^2=0$, which yields $y^*=\frac{1}{2}(\alpha+\sqrt{4v+\alpha^2})$. Notice that SOC is satisfied since $\frac{d^2 W(y)}{dy^2}|_{y=y^*}=\alpha -2y^*=-\sqrt{4v+\alpha^2}\leq 0$. The condition $y<\sqrt{2v}$ requires $\frac{1}{2}(\alpha+\sqrt{4v+\alpha^2})<\sqrt{2v}$, which is equivalent to $\alpha<\sqrt{\frac{v}{2}}$. In other words, if $\alpha\geq\sqrt{\frac{v}{2}}$, $W(y)$ is increasing in $y$ on $[0,\sqrt{2v}]$.

If $y\geq\sqrt{2v}$, $W(y)$ is increasing in $y$ if $\frac{dW(y)}{dy}=y \left(\sqrt{\alpha ^2-2 v+y^2}-y\right)+v\geq0$, which is equivalent to $y\geq\frac{v}{\alpha}$.
Note that $\frac{d^2W(y)}{dy^2}=\frac{\left(y-\sqrt{\alpha ^2-2 v+y^2}\right)^2}{\sqrt{\alpha ^2-2 v+y^2}}\geq0$, i.e., $W(y)$ is convex in $y$. The optimal solution is then either $y^*=1$ or $y^*=\sqrt{2v}$.

When $\alpha\geq\sqrt{\frac{v}{2}}$, we have $\frac{v}{\alpha}\leq\sqrt{2v}$, so $W(y)$ is increasing in $y$ on $[0,1]$, and thus $y^{P*}=1$.

When $\alpha<\sqrt{\frac{v}{2}}$,  we have $\frac{v}{\alpha}>\sqrt{2v}$, so $W(y)$ is first increasing in $y$ until $y=\frac{1}{2}(\alpha+\sqrt{4v+\alpha^2})$ when reaching the local optimum $W(\frac{1}{2}(\alpha+\sqrt{4v+\alpha^2}))$, then decreasing until $y=\frac{v}{\alpha}$, and then again increasing in $y$. So the optimal social welfare is either $W(\frac{1}{2}(\alpha+\sqrt{4v+\alpha^2}))$ or $W(1)$, depending on which one is higher. Content moderation is only conducted when $W(\frac{1}{2}(\alpha+\sqrt{4v+\alpha^2}))>W(1)$. 
To prove the existence of $\alpha^P<\sqrt{\frac{v}{2}}$ such that content moderation is only conducted when $\alpha<{\alpha^P}$, we define  $\Delta W \equiv W(\frac{1}{2}(\alpha+\sqrt{4v+\alpha^2}))-W(1)$  and prove the following: (1) $\Delta W$ is decreasing in $\alpha$, (2) $\Delta W>0$ when $\alpha=0$, and (3) $\Delta W<0$ when $\alpha=\sqrt{\frac{v}{2}}$.

(1) $\Delta W$ is decreasing in $\alpha$ because
$$\frac{\partial\Delta W}{\partial \alpha}=\frac{1}{4} \left(\alpha  \left(5 \alpha +\sqrt{\alpha ^2+4 v}-4 \sqrt{\alpha ^2-2 v+1}\right)-2 v\right)<0.$$

(2) When $\alpha=0$, 
\begin{align*}
    \Delta W|_{\alpha=0}=&W(\sqrt{v})-W(1)\\
    = &\int_0^{\sqrt{v}} (v-\frac{1}{2}(v-x^2))  dx-\int_{\sqrt{1-2v}}^1(v-\frac{1}{2}(1-x^2))  dx\\
    =&\frac{1}{3} \left(2 v^{3/2}+\left(2 \sqrt{1-2 v}-3\right) v-\sqrt{1-2 v}+1\right).
\end{align*}
Taking derivative w.r.t. $v$ yields 
\begin{equation*}
    \frac{\partial (\Delta W|_{\alpha=0})}{\partial v}=\sqrt{v}+\sqrt{1-2 v}-1,
\end{equation*}
and setting it to zero gives that $v=0$ or $v=\frac{4}{9}$. Therefore, $\Delta W|_{\alpha=0}$ is increasing in $v$ on $v\in(0,\frac{4}{9})$ and decreasing in $v$ on $v\in(\frac{4}{9},\frac{1}{2})$. Thus, for any $v\in(0,\frac{1}{2})$ $$\Delta W|_{\alpha=0}>\min\{\Delta W|_{\alpha=0,v=0},\Delta W|_{\alpha=0,v=\frac{1}{2}}\}=\min\{0,\frac{1}{3} (\frac{3}{2}+\frac{1}{\sqrt{2}})\}=0.$$

(3) When $\alpha=\sqrt{\frac{v}{2}}$, 
\begin{align*}
    \Delta W|_{\alpha=\sqrt{\frac{v}{2}}}=&W(\sqrt{2v})-W(1)\\
    = &\int_0^{\sqrt{2v}} (\sqrt{\frac{v}{2}}x+v-\frac{1}{2}(2v-x^2))  dx-\int_{\frac{\sqrt{2-3 v}-\sqrt{v}}{\sqrt{2}}}^1(\sqrt{\frac{v}{2}}x+v-\frac{1}{2}(1-x^2))  dx\\
    =&\frac{1}{12} \left(5 \sqrt{2} v^{3/2}+3 \left(\sqrt{4-6 v}-4\right) v-2 \sqrt{4-6 v}+4\right).
\end{align*}
Taking derivative w.r.t. $v$ yields 
\begin{equation*}
    \frac{\partial (\Delta W|_{\alpha=0})}{\partial v}=\sqrt{v}+\sqrt{1-2 v}-1,
\end{equation*}
and setting it to zero gives that $v=\frac{49}{338}$ or $v=\frac{1}{2}$. Therefore, $\Delta W|_{\alpha=\sqrt{\frac{v}{2}}}$ is decreasing in $v$ on $v\in(0,\frac{49}{338})$ and increasing in $v$ on $v\in(\frac{49}{338},\frac{1}{2})$. Thus, for any $v\in(0,\frac{1}{2})$ $$\Delta W|_{\alpha=\sqrt{\frac{v}{2}}}<\max\{\Delta W|_{\alpha=\sqrt{\frac{v}{2}},v=0},\Delta W|_{\alpha=\sqrt{\frac{v}{2}},v=\frac{1}{2}}\}=\max\{0,0\}=0.$$

Therefore, we claim that there exists $\alpha^P<\sqrt{\frac{v}{2}}$ such that $y^{P*}=\frac{1}{2}(\alpha+\sqrt{4v+\alpha^2})$ if $\alpha<\alpha^P$ while no content moderation is conducted ($y^{P*}=1$) otherwise. 

Since $\alpha^P$, $\alpha^S$, and $\alpha^A$ are all single-variable functions of $v$, one can easily check their relative sizes numerically. The result is $\alpha^P<\alpha^S<\alpha^A$ for all $v\in(0,\frac{1}{2})$.

When the social planner moderates content ($\alpha<\alpha^P$), we have 
\begin{align*}
    y^{S*}=  \sqrt{2v/3} < \sqrt{v} 
      < y^{P*}=\frac{1}{2}(\alpha+\sqrt{4v+\alpha^2}) 
      <\frac{1}{2}(\sqrt{\frac{v}{2}}+\sqrt{4v+(\sqrt{\frac{v}{2}})^2})=\sqrt{2v}=y^{A*},
\end{align*}
where the last inequality comes from $\alpha<\alpha^P<\sqrt{\frac{v}{2}}$.
This completes the proof. 

A final check about the claim that ``the social planner only prunes users with a negative utility contribution to the society'' (on page \pageref{check_u_contrib}) is as follows.\label{app:check_u_contrib} Consider the interior solution $y^{P*}=\frac{1}{2}(\alpha+\sqrt{4v+\alpha^2})$, the net utility contribution of a user at $y^{P*}$ is $\alpha y^{P*}+v-\frac{1}{2}(y^{P*})^2$ where the last term is the total negative utility this user imposes to all other users on the platform. Substituting $y^{P*}$ with $\frac{1}{2}(\alpha+\sqrt{4v+\alpha^2})$, one can find that this net utility contribution is exactly zero.

\end{proof}

\bigskip

\begin{proof}[\bf Proof of Proposition \ref{prop:social_optimal_k}]\label{pf:prop8}

We show that for any $k\in[0,\frac{1}{2}]$, the optimal social welfare $W^*_{k}$ is weakly increasing in $k$ numerically by exhausting the parameter space of $\alpha\in[0,1]$, $v\in[0,\frac{1}{2}]$, and $c\in[v,\alpha+2v]$ with an increment of 0.05. Details are provided in Appendix \ref{appdx:numerical}, especially Section \ref{app:numerical_subsc}.

Figure \ref{fig:sw_k} below illustrates the relationship between $k$ and $W^*_{k}$ when $\alpha=0.2$, $v=0.25$, and $c=0.3$. A similar pattern can be seen for all other  combinations of $(\alpha,v,c)$.

\begin{figure}[h]
\caption{Optimal social welfare $W^*_{k}$ and technology accuracy $k$ ($v=0.25, \alpha=0.2$, $c=0.3$)}\label{fig:sw_k}
  \centering
   \includegraphics[width=0.5\textwidth]{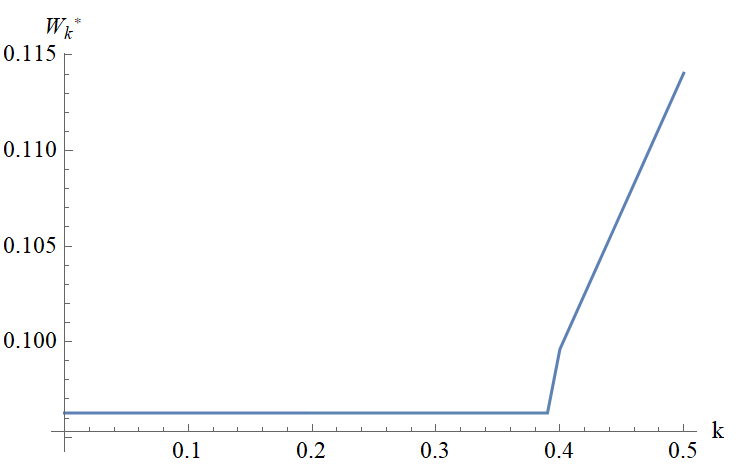}
\end{figure}

\end{proof}

\section{Equilibrium with Imperfect Technology}\label{appdx:impeft_ai}

\subsection{Advertising Revenue Model}\label{appdx:impeft_ai_ad}

Lemma \ref{lem:config_imperfectAI} characterizes the user equilibrium under advertising, given a content moderation policy $y$.

\begin{lemma}\label{lem:config_imperfectAI}
For an ad-supported platform, when $k\in[0,\frac{1}{2}]$ and the platform does content moderation, there exists $x^A_{2,k}\in[0,1]$ such that all users in range $[x^A_{2,k},1]$ participate in the platform. In particular, 
\begin{equation}\label{eq:app_xa2k}
   x^A_{2,k} = \min\{1,\sqrt{\alpha^2+1+\frac{2c(1+2k)-4v}{1-2k}}-\alpha\}\equiv
    \begin{cases}
    \sqrt{\alpha^2+1+\frac{2c(1+2k)-4v}{1-2k}}-\alpha & \text{if $k<\overline{k}$;} \\
    1 & \text{if $k\geq\overline{k}$,} 
  \end{cases}
\end{equation}
where $\overline{k}=\frac{\alpha+2v-c}{2(\alpha+c)}$. 
For any $y\in[0,1]$
, there exists $x^A_{1,k}(y)\in[0,y]$ such that if $y<x^A_{2,k}$, the user set of the platform $\mathcal{X}^A$ is $[x^A_{1,k}(y),y]\cup[x^A_{2,k},1]$; if $y\geq x^A_{2,k}$, $\mathcal{X}^A$ is $[x^A_{1,k}(y),1]$. In particular,
\begin{equation}\label{eq:app_lemma2}
   x^A_{1,k}(y) = 
    \begin{cases}
    \sqrt{\alpha^2+\max\big\{0,y^2+\min\{2\alpha y,\frac{(1-2k)(2c+1-(x^A_{2,k})^2)-4v}{1+2k}\}\big\}}-\alpha & \text{if $y<x^A_{2,k}$;} \\
    \sqrt{\alpha^2+\max\big\{0,y^2+\min\{2\alpha y,\frac{(1-2k)(2c+1-y^2)-4v}{1+2k}\}\big\}}-\alpha & \text{if $y\geq x^A_{2,k}$.}
  \end{cases}
\end{equation}

Furthermore, $x_{2,k}^A$ is increasing in $k$ and $x_{1,k}^A(y)$ is decreasing in $k$.

\end{lemma}

\begin{proof}[\bf Proof of Lemma \ref{lem:config_imperfectAI}]

Based on Equations (\ref{eq:qk}) and (\ref{eq:utility_fcn_2}), we have
\begin{equation}
    U(x) = 
    \begin{cases}
    \alpha x(\frac{1}{2}+k)-c(\frac{1}{2}-k)+v-\int_{\tilde{x}\in\hat{\mathcal{X}},x<\tilde{x}\leq{y}}\tilde{x}(\frac{1}{2}+k)d\tilde{x}-\int_{\tilde{x}\in\hat{\mathcal{X}},\tilde{x}>{y}}\tilde{x}(\frac{1}{2}-k)d\tilde{x} & \text{if $x
    \leq y$;} \\
    \alpha x(\frac{1}{2}-k)-c(\frac{1}{2}+k)+v-\int_{\tilde{x}\in\hat{\mathcal{X}},\tilde{x}>{x}}\tilde{x}(\frac{1}{2}-k)d\tilde{x} & \text{if $x>y$.} 
  \end{cases}
\end{equation}

Note that there is a discontinuity at $x=y$: $U(y^-)$>$U(y^+)$ as long as $k>0$. Also, similar to what is shown in the proof of Lemma \ref{lem:config}, $U(x)$ is increasing in $x$ on $[0,y]$ and also increasing in $x$ on $(y,1]$. Therefore, there can be possibly two segments of participating users with $U(x)>0$:  the moderate users $[x_{1,k}^A,y]$ and the extreme users  $[x_{2,k}^A,1]$.

First consider the extreme users in $(y,1]$. 
If $U(1)<0$,  no users in $(y,1]$ participate in the platform since $U(x)$ is increasing in $x$. $U(1)=\alpha(\frac{1}{2}-k)-c(\frac{1}{2}+k)+v<0$ is equivalent to $k>\frac{\alpha+2v-c}{2(\alpha+c)}=\overline{k}$. Otherwise, when $k\leq\overline{k}$, all users in $[x_{2,k}^A,1]$ participate, where $x_{2,k}^A$ solves $U(x_{2,k}^A)=\alpha x_{2,k}^A(\frac{1}{2}-k)-c(\frac{1}{2}+k)+v-\int_{x_{2,k}^A}^1\tilde{x}d\tilde{x}=0$, which gives $x_{2,k}^A=\sqrt{\alpha^2+1+\frac{2c(1+2k)-4v}{1-2k}}-\alpha$. Therefore, Equation (\ref{eq:app_xa2k}) holds.

For the moderate users in $[0,y]$, $x_{1,k}^A$ can be given by the condition $U(x_{1,k}^A)=0$ if the solution to this condition is interior ($0<x_{1,k}^A<y$). Denote the interior solution as $\tilde{x_{1,k}^A}$. Depending on whether $y<x_{2,k}^A$ or $y\geq x_{2,k}^A$, the condition $U(\tilde{x_{1,k}^A})=0$  is given by
\begin{equation}\label{eq:app_x1k}
    \begin{cases}
   \alpha \tilde{x_{1,k}^A}(\frac{1}{2}+k)-c(\frac{1}{2}-k)+v-\int_{\tilde{x_{1,k}^A}}^{y}\tilde{x}(\frac{1}{2}+k)d\tilde{x}-\int_{x_{2,k}^A}^1\tilde{x}(\frac{1}{2}-k)d\tilde{x}=0 & \text{if $y<x^A_{2,k}$;} \\
    \alpha \tilde{x_{1,k}^A}(\frac{1}{2}+k)-c(\frac{1}{2}-k)+v-\int_{\tilde{x_{1,k}^A}}^{y}\tilde{x}(\frac{1}{2}+k)d\tilde{x}-\int_{y}^1\tilde{x}(\frac{1}{2}-k)d\tilde{x}=0 & \text{if $y\geq x^A_{2,k}$.}
  \end{cases}
\end{equation}
Solving Equation (\ref{eq:app_x1k}), we have
\begin{equation}\label{eq:app_xiksln}
 \tilde{x_{1,k}^A} = 
    \begin{cases}
    \sqrt{\alpha^2+y^2+\frac{(1-2k)(2c+1-(x^A_{2,k})^2)-4v}{1+2k}}-\alpha & \text{if $y<x^A_{2,k}$;} \\
    \sqrt{\alpha^2+y^2+\frac{(1-2k)(2c+1-y^2)-4v}{1+2k}}-\alpha & \text{if $y\geq x^A_{2,k}$.}
  \end{cases}
\end{equation}
If $0<\tilde{x_{1,k}^A}<y$, $x_{1,k}^A=\tilde{x_{1,k}^A}$. If $\tilde{x_{1,k}^A}\leq0$ or $\tilde{x_{1,k}^A}\geq y$, corner solutions apply. I.e.,
\begin{equation}\label{eq:app_xiksln2}
{x_{1,k}^A(y)} = 
    \begin{cases}
    0 & \text{if $\tilde{x_{1,k}^A}\leq0$;} \\
    \tilde{x_{1,k}^A} & \text{if $0<\tilde{x_{1,k}^A}<y$;} \\
    y & \text{if $\tilde{x_{1,k}^A}\geq y$.}
  \end{cases}
\end{equation}
Rewriting Equations (\ref{eq:app_xiksln}) and (\ref{eq:app_xiksln2}) in a dense format obtains Equation (\ref{eq:app_lemma2}).

With Equations (\ref{eq:app_xa2k}) and (\ref{eq:app_lemma2}), we see immediately that  $x_{2,k}^A$ is increasing in $k$ and $x_{1,k}^A(y)$ is decreasing in $k$, since $(1-2k)$ is decreasing in $k$ and $(1+2k)$ is increasing in $k$.

\end{proof}

The following lemma further investigates the platform's optimal content moderation strategy given users' response.

\begin{lemma}\label{lem:cm_k}
{\bf \em } 
Let $\hat{y_k^A}\equiv\sqrt{\max\{0,\frac{4v-(1-2k)(2c+1-(x_{2,k}^A)^2)}{1+2k}\}}$. The optimal level of content moderation under advertising revenues ($y^{A*}_k$) can be characterized by the following:
\begin{enumerate}[label=\bfseries Step \arabic*:,leftmargin=*,labelindent=1em]
    \item [Case 1.] If $k\geq\overline{k}$, then $x_{2,k}^A=1$, and $y^{A*}_k$ is (a) $1$ if $k<k_1^A$, and (b) $\hat{y_k^A}$ if $k\geq k_1^A$. 
    \item [Case 2.] If $k<\overline{k}$ and $\hat{y_k^A}<x_{2,k}^A$, then $x_{2,k}^A<1$, and $y^{A*}_k$ is (a) $x_{2,k}^A$ if $k<k_2^A$, and (b) $\hat{y_k^A}$ if $k\geq k_2^A$. 
    \item [Case 3.] If $k<\overline{k}$ and $\hat{y_k^A}\geq x_{2,k}^A$, then the market can be fully covered with any $y^{A*}_k\in [x_{2,k}^A,\hat{y_k^A}]$,
\end{enumerate}
where $k_1^A, k_2^A$ are constant.

\end{lemma}

In Case 3, there are multiple maximizers. As a tie-breaking rule, we assume that the platform will choose the lowest $y^{A*}_k=x_{2,k}^A$ (the most strict policy) to make the platform as moderate as possible.

\begin{proof}[\bf Proof of Lemma \ref{lem:cm_k}]

Since Case 1 ($x_{2,k}^A=1$) is just a special case of Case 2, we only need to show the following: When $\hat{y_k^A}<x_{2,k}^A$, $y^{A*}_k=x_{2,k}^A$ if $k<k_0$ and $y^{A*}_k=\hat{y_k^A}$ if $k>k_0$, where $k_0$ is a constant; when $\hat{y_k^A}\geq x_{2,k}^A$,  any $y^{A*}_k\in [x_{2,k}^A,\hat{y_k^A}]$ can make the market fully covered ($X^{A*}=1$).

If $y> x_{2,k}^A$, then
\begin{align*}
    x_{1,k}^A&=\sqrt{\alpha^2+\max\{0,y^2+\min\{2\alpha y,\frac{(1-2k)(2c+1-y^2)-4v}{1+2k}\}\}}-\alpha\\
    &=\sqrt{\alpha^2+\max\{0,\min\{y(y+2\alpha),\frac{(2c+1)(1-2k)+4ky^2-4v}{1+2k}\}\}}-\alpha
\end{align*}
is increasing in $y$, and thus the user base $1-x_{1,k}^A$ is decreasing in $y$, so any $y>x_{2,k}^A$ cannot be an optimal choice.

If $y<x_{2,k}^A$, note that $\tilde{x_{1,k}^A}\geq0$ is equivalent to $y\geq \hat{y_k^A}$ where $\hat{y_k^A}=\sqrt{\frac{4v-(1-2k)(2c+1-(x_{2,k}^A)^2)}{1+2k}}$ is determined through solving for $y$ from $\tilde{x_{1,k}^A}=0$ (which is also equivalent to $U(0)=0$, by definition of $\tilde{x_{1,k}^A}$).\footnote{If there is no real number solution to this equation,  we set $\hat{y_k^A}=0$ without loss of generality.}
Any $y<\hat{y_k^A}$ also cannot be an optimal choice because if $y<\hat{y_k^A}$ then ${x_{1,k}^A}=0$ so the user base is just $y+1-x_{2,k}^A$ which is increasing in $y$.

Therefore, we only consider $y\in[\hat{y_k^A},x_{2,k}^A]$. In this case,
$\frac{\partial \pi^{A}_k}{\partial y}=\frac{\partial \zeta(y-x_{1,k}^A(y)+1-x_{2,k}^A)}{\partial y}=\zeta\Big(1-\frac{y}{\sqrt{\alpha^2+y^2+\frac{(1-2k)(2c+1-(x^A_{2,k})^2)-4v}{1+2k}}}\Big)=:\zeta(1-\frac{y}{L(k)})$, where $$L(k)\equiv\sqrt{\alpha^2+y^2+\frac{(1-2k)(2c+1-(x^A_{2,k})^2)-4v}{1+2k}}.$$
Note that $x_{2,k}^A$ is increasing in $k$ and thus $L(k)$ is decreasing in $k$. Therefore, $\frac{\partial \pi^{A}_k}{\partial y}$ is decreasing in $k$. 
Furthermore, $\frac{\partial \pi^{A}_k}{\partial y}|_{k=\frac{1}{2}}<0$ because $L(\frac{1}{2})=\sqrt{y^2+\alpha^2-2v}<y$ since $\alpha<\alpha^S<\alpha^A=\sqrt{2v}$.
 $\frac{\partial \pi^{A}_k}{\partial y}|_{k=0}>0$ because
\begin{align*}
    L(0)&=\sqrt{y^2+\alpha^2+2c+1-4v-(x_{2,k=0}^A)^2}\\
    &=\sqrt{y^2+\alpha^2+1+2c-4v-(\sqrt{\alpha^2+1+2c-4v}-\alpha)^2}\\
    &>\sqrt{y^2+\alpha^2+1+2c-4v-(\sqrt{\alpha^2+1+2c-4v})^2}\\
    &=y.
\end{align*} 
So there exists $k_0$ such that $\frac{\partial \pi^{A}_k}{\partial y}>0$ when $k<k_0$ and $\frac{\partial \pi^{A}_k}{\partial y}<0$ when $k>k_0$. Since $y\in[\hat{y_k^A},x_{2,k}^A]$, we know that the optimal $y^{A*}_k=x_{2,k}^A$ when $k<k_0$ and $y^{A*}_k=\hat{y_k^A}$ when $k>k_0$.

Note that if $\hat{y_k^A}>x_{2,k}^A$, it simply means that the market can be fully covered by choosing any $y^{A*}_k\in[x_{2,k}^A,\hat{y_k^A}]$. This is because every user with $x\leq y$ participates when $y\leq\hat{y_k^A}$, and every user with $x\geq x_{2,k}^A$ participates regardless of the choice of $y$.

\end{proof}

Lemma \ref{lem:cm_k} indicates that unless the market is fully covered, there are generally two potential levels of content moderation that the platform can choose. As the  technology becomes more accurate, the platform tends to choose the higher level of content moderation (a smaller $y$).\footnote{Note that this statement is about the choice between two levels for a given $k$ and that it does not mean $y^{A*}$ is decreasing in $k$.} Meanwhile, it points out the possibility that the market is fully covered when the technology is imperfect, and thus the possibility that an imperfect technology may enlarge the market for the platform. 

\subsection{Subscription Revenue Model}\label{appdx:impeft_ai_subsc}

Similar to the advertising case, Lemma \ref{lem:config_imperfectAI_subsc} gives the full characterization of the user equilibrium under advertising, given a content moderation policy $y$ and subscription fee $p$. 

\begin{lemma}\label{lem:config_imperfectAI_subsc}
 Suppose the subscription fee $p$ is given. For a subscription-supported platform, when $k\in[0,\frac{1}{2}]$ and the platform does content moderation, there exists $x^S_{2,k}(p)\in(0,1]$ such that all users in $[x^S_{2,k}(p),1]$ participate in the platform. In particular, 
\begin{equation}\label{app:x2ksp}
   x^S_{2,k}(p) = \min\{1,\sqrt{\alpha^2+1+\frac{2c(1+2k)-4(v-p)}{1-2k}}-\alpha\}.
\end{equation}
Furthermore, for any $y\in[0,1]$, there exists $x^S_{1,k}(y,p)\in[0,y]$ such that if $y<x^S_{2,k}(p)$, the user set of the platform $\mathcal{X}^S$ is $[x^S_{1,k}(y,p),y]\cup[x^S_{2,k}(p),1]$; if $y\geq x^S_{2,k}(p)$, $\mathcal{X}^S$ is $[x^S_{1,k}(y,p),1]$. In particular,
\begin{equation}\label{app:x1ksp}
   x^S_{1,k}(y,p) = 
    \begin{cases}
    \sqrt{\alpha^2+\max\big\{0,y^2+\min\{2\alpha y,\frac{(1-2k)(2c+1-(x^S_{2,k}(p))^2)-4(v-p)}{1+2k}\}\big\}}-\alpha & \text{if $y<x^S_{2,k}(p)$;} \\
    \sqrt{\alpha^2+\max\big\{0,y^2+\min\{2\alpha y,\frac{(1-2k)(2c+1-y^2)-4(v-p)}{1+2k}\}\big\}}-\alpha & \text{if $y\geq x^S_{2,k}(p)$.}
  \end{cases}
\end{equation}
\end{lemma}

\begin{proof}[\bf Proof of Lemma \ref{lem:config_imperfectAI_subsc}]
The proof of Lemma \ref{lem:config_imperfectAI_subsc} is the same as that of Lemma \ref{lem:config_imperfectAI} except that we change $v$ to $v-p$.
\end{proof}

Given the response of users, we next derive the optimal content moderation policy $y^{S*}$ and pricing $p^{S*}$. 

The following lemma describes the optimal level of content moderation, $y^{S*}(p)$, for any given $p$.

\begin{lemma}\label{lem:imperfectAI_y_eq}
Let $\hat{y_k^S}(p)\equiv\sqrt{\max\{0,{\frac{4(v-p)-(1-2k)(2c+1-(x_{2,k}^S(p))^2)}{1+2k}}\}}$. The optimal level of content moderation under subscription revenues ($y^{S*}_k(p)$) can be characterized by the following:
\begin{enumerate}[label=\bfseries Step \arabic*:,leftmargin=*,labelindent=1em]
    \item [Case 1.] If $p>p_{1,k}$, then $x_{2,k}^S(p)=1$, and $y^{S*}_k(p)$ is (a) $1$ if $k<k_1^S$, and (b) $\hat{y_k^S}(p)$ if $k\geq k_1^S$. 
    \item [Case 2.] If $p_{2,k}<p\leq p_{1,k}$, then $x_{2,k}^S(p)<1$, and $y^{S*}_k(p)$ is (a) $x_{2,k}^S(p)$ if $k<k_2^S$, and (b) $\hat{y_k^S}(p)$ if $k\geq k_2^S$. 
    \item [Case 3.] If $p\leq p_{2,k}$, then the market can be fully covered with any $y^{S*}_k(p)\in [x_{2,k}^S(p),\hat{y_k^S}(p)]$.
\end{enumerate}
where $k_1^S, k_2^S$ are constants, and $p_{1,k},p_{2,k}$ are constants when $k$ is given.
\end{lemma}

\begin{proof}[\bf Proof of Lemma \ref{lem:imperfectAI_y_eq}]
The proof can be built on that of Lemma \ref{lem:cm_k}. The only difference here is that $p$ is another decision variable of the platform, so whether $x_{2,k}^S(p)<1$ depends not only on the technology accuracy $k$ but also on the pricing of the platform. Therefore, only the condition for each of the three cases will change.

If $x_{2,k}^S(p)=1$, which is equivalent to $\sqrt{\alpha^2+1+\frac{2c(1+2k)-4(v-p)}{1-2k}}-\alpha>1\Leftrightarrow p> v+\alpha(\frac{1}{2}-k)-c(\frac{1}{2}+k)$, no users with $x>y$ participate in the platform. Let
\begin{equation}\label{eq:app_p1k}
    p_{1,k}\equiv v+\alpha(\frac{1}{2}-k)-c(\frac{1}{2}+k),
\end{equation}
and we have found the condition for Case 1.

If  $x_{2,k}^S(p)<1$, let $\hat{y_k^S}(p)$ be the content moderation strategy such that $U(0)-p=0$. Solving it, we have $\hat{y_k^S}(p)=\sqrt{\frac{4(v-p)-(1-2k)(2c+1-(x_{2,k}^S(p))^2)}{1+2k}}$.\footnote{If there is no real number solution to this equation,  we set $\hat{y_k^S}(p)=0$ without loss of generality.} We then only need to check whether $\hat{y_k^S}(p)<x_{2,k}^S(p)$. If so, it corresponds to Case 2; otherwise, it corresponds to Case 3. By Equation (\ref{app:x2ksp}), we know that when $x_{2,k}^S(p)<1$, $\frac{dx_{2,k}^S(p)}{dp}=\frac{2/(1-2k)}{\sqrt{\alpha^2+1+\frac{2c(1+2k)-4(v-p)}{1-2k}}}=\frac{2/(1-2k)}{dx_{2,k}^S(p)+\alpha}>0$ so $x_{2,k}^S(p)$ is increasing in $p$. Furthermore, $\sign(\frac{d\hat{y_k^S}(p)}{dp})=\sign(\frac{d[4(v-p)-(1-2k)(2c+1-(x_{2,k}^S(p))^2)]}{dp})=\sign(-4+(1-2k)2x_{2,k}^S(p)\frac{dx_{2,k}^S(p)}{dp})=\sign(-4+4\frac{x_{2,k}^S(p)}{x_{2,k}^S(p)+\alpha})<0$, so  $\hat{y_k^S}(p)$ is decreasing in $p$. Therefore, the condition $\hat{y_k^S}(p)\lessgtr x_{2,k}^S(p)$ is equivalent to $p\gtrless p_{2,k}$, which gives the conditions for Cases 2 and 3. Solving for $p_{2,k}$ by setting $\hat{y_k^S}(p_{2,k})= x_{2,k}^S(p_{2,k})$, we have 
\begin{equation}\label{eq:app_p2k}
    p_{2,k}=v-\frac{1}{4}-\frac{c \left(12 k^2+1\right)}{2 (2 k+1)}+\frac{2 \alpha  (1-2 k) k \sqrt{8 c k (2 k+1)+\alpha ^2 (1-2 k)^2}}{(2 k+1)^2}+\frac{1}{2} k \left(1-4\alpha ^2 \left(1-\frac{8 k}{(2 k+1)^2}\right)\right).
\end{equation}

\end{proof}

We can see from Lemma \ref{lem:imperfectAI_y_eq} that similar to the advertising model case, there are two potential levels of content moderation and the one with more moderation (smaller $y$) is more preferred as $k$ increases. Also, we can see that when price is too high, there are no users more extreme than  $y$ who participate in the platform. Denote the associated profits in Cases 1(a), 1(b), 2(a), 2(b), and 3 as $\pi_k^{1a}(p)$, $\pi_k^{1b}(p)$, $\pi_k^{2a}(p)$, $\pi_k^{2b}(p)$, and $\pi_k^{3}(p)$, respectively. It is clear that
\begin{equation}\label{eq:global_max_kp}
    \pi_k^{S*}=\max\{ \max_{p>p_{1,k}}\pi_k^{1a}(p), \max_{p>p_{1,k}}\pi_k^{1b}(p), 
    \max_{p_{2,k}<p\leq p_{1,k}}\pi_k^{2a}(p), 
    \max_{p_{2,k}<p\leq p_{1,k}}\pi_k^{2b}(p), 
    \max_{p\leq p_{2,k}}\pi_k^{3}(p)\}.
\end{equation}

\section{\large Numerically Solving for Equilibrium in Imperfect Technology Case}\label{appdx:numerical}

Lemmas \ref{lem:config_imperfectAI} to \ref{lem:imperfectAI_y_eq} 
in Appendix \ref{appdx:impeft_ai} give the analytical characterization of the equilibrium in the imperfect technology case. However, to fully solve the equilibrium (especially for the subscription case) and to carry out any further analysis based on the equilibrium are  analytically challenging, since the expressions are so complicated that only implicit functions can be provided for equilibrium characterization. Note that the range of each exogenous variable in our model  is bounded ($\alpha\in[0,1]$, $v\in[0,\frac{1}{2}]$, and $c\in[v,\alpha+2v]\subset[0,2]$), so we can numerically compute the results exhaustively. For each variable, we discretize the range with a grid of 0.05, and enumerate all possible values within the given range. In other words, for $\alpha$ and $v$, we have values of 0, 0.05, 0.1, 0.15, ..., 0.95, and 1; for $c$, we have 0, 0.05, ..., 1.95, and 2. For each combination of the parameters, which is denoted as a tuple $(\alpha,v,c)$, we further check whether the following two conditions are satisfied: (1) $v\leq c\leq\alpha+2v$ and (2) $\alpha\leq\alpha^S$, which is equivalent to $\pi_S(y=\sqrt{\frac{2v}{3}})\geq\pi_S(y=1)$. We only proceed if both of them are satisfied.

Then, for any given tuple $(\alpha,v,c)$, we can find out the optimal $y^{A*}$ (in the advertising case) or $y^{S*}$ and $p^*$ (in the subscription case) as a function of $k$. For $k$, we also enumerate all possible values between $[0,\frac{1}{2}]$, with a grid of 0.01. To make the search algorithm more efficient, we leverage the analytical results in Appendix \ref{appdx:impeft_ai}. 

\subsection{Advertising}\label{app:numerical_ad}

For the advertising case, given any $(\alpha,v,c)$ and $k$, we can directly get the optimal $y^{A*}_k$ based on the analytical solution provided by Lemma \ref{lem:cm_k}.
Based on Lemma \ref{lem:cm_k}, we know that the optimal content moderation policy $y^{A*}_k$ is either 
$$y^{A*}_k=x^A_{2,k} \equiv \min\{1,\sqrt{\alpha^2+1+\frac{2c(1+2k)-4v}{1-2k}}-\alpha\},$$
or
$$y^{A*}_k=\hat{y_k^A}\equiv\sqrt{\max\{0,\frac{4v-(1-2k)(2c+1-(x_{2,k}^A)^2)}{1+2k}\}},$$
whichever gives the largest user base. The user base is calculated as $X^{A*}\equiv 1-x^A_{2,k}+y^{A*}-x^A_{1,k}(y^{A*})$, where the expression for $x^A_{1,k}(y)$ is given by Equation (\ref{eq:app_lemma2}).\footnote{In Case 3 of Lemma \ref{lem:cm_k}, there are multiple maximizers. Based on the tie-breaking rule, we know that if $x_{2,k}^A$ and $\hat{y_k^A}$ induce the same profit for the platform, it chooses $y^{A*}_k=\min\{x_{2,k}^A,\hat{y_k^A}\}$.}
We thus also find out the equilibrium user configuration (i.e., $x_{1,k}^A$ and $x_{2,k}^A$) and the equilibrium profit (user base) of the platform when it conducts content moderation. Note that we also need to compare this optimal user base ($X^{A*}_k$) when conducting content moderation with the user base when no moderation is conducted at all (denoted as $X^A_0$), to see whether the equilibrium strategy is $y^{A*}_k$ or simply no content moderation. The technology accuracy does not matter when no content moderation is conducted, so the user base $$X^A_0=1+\alpha-\sqrt{\alpha^2+1-2v},$$
which is given in Proposition \ref{prop:ad_eq}.

So far, we have solved the imperfect technology equilibrium for a platform under advertising, and also calculated the equilibrium quantities.

\subsection{Subscription}\label{app:numerical_subsc}

For the subscription case, given any $(\alpha,v,c)$ and $k$, to solve for the equilibrium needs more work. Lemma \ref{lem:imperfectAI_y_eq} only provides a partial equilibrium for any given subscription fee $p$. To find out the optimal $p^*_k$ as well as the associated $y^{S*}_k$, we do the following numerical analysis. 

There are 5 subcases in Lemma \ref{lem:imperfectAI_y_eq} (Cases 1(a), 1(b), 2(a), 2(b), and 3). We numerically solve a constrained maximization problem over $p$ for each subcase. Let 
$\pi_k^{1a}(p)$, $\pi_k^{1b}(p)$, $\pi_k^{2a}(p)$, $\pi_k^{2b}(p)$, and $\pi_k^{3}(p)$ denote the associated profits in Cases 1(a), 1(b), 2(a), 2(b), and 3, respectively. Based  Lemma \ref{lem:imperfectAI_y_eq}, we can write out the expressions for the profit (objective function) as well as the constraints in each subcase:
\begin{itemize}
    \item Case 1(a): 
\begin{itemize}
    \item Objective function: $\pi_k^{1a}(p)=p(1-x_{1,k}^S(1,p)).$
    \item Constraints: $p>  p_{1,k}$,  $p>0$.
\end{itemize}
    \item Case 1(b): 
\begin{itemize}
    \item Objective function: $\pi_k^{1b}(p)=p\hat{y_k^S}(p)=p\sqrt{\max\{0,{\frac{4(v-p)-(1-2k)2c}{1+2k}}\}}.$
    \item Constraints: $p> p_{1,k}$, $p>0$.
\end{itemize}
    \item Case 2(a): 
\begin{itemize}
    \item Objective function: $\pi_k^{2a}(p)=p(1-x_{1,k}^S(x_{2,k}^S,p)).$
    \item Constraints: $p\leq  p_{1,k}$, $p>p_{2,k}$, $p>0$.
\end{itemize}
    \item Case 2(b): 
\begin{itemize}
    \item Objective function: $\pi_k^{2b}(p)=p(\hat{y_k^S}(p)+1-x_{2,k}^S(p))=p(\sqrt{\max\{0,{\frac{4(v-p)-(1-2k)(2c+1-(x_{2,k}^S(p))^2)}{1+2k}}\}}+1-x_{2,k}^S(p)).$
    \item Constraints: $p\leq  p_{1,k}$, $p>p_{2,k}$, $p>0$.
\end{itemize}
    \item Case 3: 
\begin{itemize}
    \item Objective function: $\pi_k^{3}(p)=p.$
    \item Constraints:  $p\leq p_{2,k}$, $p>0$.
\end{itemize}
\end{itemize}
where the expressions for $x_{2,k}^S(p)$, $x_{1,k}^S(y,p)$,  $p_{1,k}$ and $p_{2,k}$ are given by Equations (\ref{app:x2ksp}), (\ref{app:x1ksp}), (\ref{eq:app_p1k}) and (\ref{eq:app_p2k}), respectively. Each optimization problem is solve by a direct search algorithm, which is implemented by the \texttt{NMaximize} function in Mathematica.\footnote{\texttt{NMaximize} function in Mathematica uses one of the four direct search algorithms (Nelder-Mead, differential evolution, simulated annealing, and random search), then fine-tunes the solution by using a combination of KKT solution, the interior point, and a penalty method. Source: \url{https://reference.wolfram.com/language/tutorial/ConstrainedOptimizationComparison.html}.}

Then we find the maximum across all the subcases (see Equation (\ref{eq:global_max_kp})), which gives the optimal $y^{S*}$ and $p^*$, as well as the optimal profit when the platform conducts content moderation. Again, we compare this profit with the profit when no moderation is conducted. If the latter is larger, no content moderation is conducted in equilibrium. Based on Proposition \ref{prop:subsc_eq}, we know that the optimal profit when no content moderation is conducted, $\pi_0^S$, is given by $$\pi^{S}_0=p_1^*(1+\alpha-\sqrt{\alpha^2+1-2(v-p_1^*)}),$$
where $p_1^*\equiv\frac{1}{9}\Big[(1+\alpha)\sqrt{2(2-3v+2\alpha^2+\alpha)}-2(1-3v+\alpha^2-\alpha)\Big]$.

\subsection{Social Planner}\label{app:numerical_social_pl}

For the social planner case, users' response to a given content moderation $y$ is the same to that on a platform under advertising, so all the results in Lemma \ref{lem:config_imperfectAI} hold for the social planner. In other words, the marginal users $x_{2,k}^P\equiv x_{2,k}^A$ and $x_{1,k}^P(y)\equiv x_{1,k}^A(y)$, where the expressions are given in Equations (\ref{eq:app_xa2k}) and (\ref{eq:app_lemma2}). The only difference for a social planner is the objective function, which is no longer the size of user base, but the total welfare of the users, which we denote as $W_k(y)$. As an extension of Equation (\ref{eq:welfare}) to the imperfect technology case, we have
\begin{align*}
    W_k(y) = 
    \begin{cases}
    \begin{aligned}
    \int_{x_{1,k}^P(y)}^y\Big(    (\frac{1}{2}+k) \alpha  x -(\frac{1}{2}-k) c +v - (\frac{1}{2}+k)\frac{1}{2}(y^2-x^2)- (\frac{1}{2}-k)\frac{1}{2}(1-(x_{2,k}^P)^2)\Big)dx\\
    + \int_{x_{2,k}^P}^1\Big(  (\frac{1}{2}-k) \alpha  x -(\frac{1}{2}+k) c +v - (\frac{1}{2}-k)\frac{1}{2}(1-x^2)\Big)dx \text{\quad\quad if $y<x_{2,k}^P$,}
\end{aligned}\\
    \begin{aligned}
    \int_{x_{1,k}^P(y)}^y\Big(    (\frac{1}{2}+k) \alpha  x -(\frac{1}{2}-k) c +v - (\frac{1}{2}+k)\frac{1}{2}(y^2-x^2)- (\frac{1}{2}-k)\frac{1}{2}(1-y^2)\Big)dx\\
    + \int_y^1\Big(  (\frac{1}{2}-k) \alpha  x -(\frac{1}{2}+k) c +v - (\frac{1}{2}-k)\frac{1}{2}(1-x^2)\Big)dx \text{\quad\quad if $y\geq x_{2,k}^P$.}
\end{aligned}
    \end{cases}
\end{align*}
The optimal moderation strategy under imperfect technology $y_k^{P*}\in[0,1]$ maximizes $W_k(y)$. Given any $(\alpha,v,c)$ and $k$, the optimization problem is solve by a direct search algorithm, which is implemented by the \texttt{NMaximize} function in Mathematica. 

Note that we also need to compare this optimal optimal social welfare when conducting content moderation with the social welfare when no moderation is conducted at all (denoted as $W_0$), to see whether the equilibrium strategy is $y^{P*}_k$ or simply no content moderation. The social welfare without content moderation is given by
\begin{align*}
W_0=&\int_{\sqrt{\alpha ^2-2 v+1}-\alpha }^1 \left(v-\frac{1}{2} \left(1-x^2\right)+\alpha  x\right)  dx\\
    = & \frac{1}{3} \left(\alpha ^2 \left(\sqrt{\alpha ^2-2 v+1}-\alpha \right)+\sqrt{\alpha ^2-2 v+1}+v \left(3 \alpha -2 \sqrt{\alpha ^2-2 v+1}+3\right)-1\right),
\end{align*}
which comes from Equation (\ref{eq:welfare}) by plugging in $y=1$.

\bigskip\bigskip\bigskip

So far, under each revenue model or the social planner's problem, we have now developed a dataframe where each row corresponds to a combination of $\alpha$, $v$, $c$, and $k$. The columns in this dataframe document the equilibrium content moderation strategy $y_k^*$, the equilibrium profit $\pi_k^*$ (or social welfare $W_k^*$), and the equilibrium user configuration (the marginal users $x_{1,k}^{*}$ and $x_{2,k}^{*}$). In other words, we obtain three dataframes with the following attributes\label{tables}:
\begin{itemize}
    \item Dataframe A (advertising): $\alpha$, $v$, $c$, $k$, $y_k^{A*}$, $\pi_k^{A*}$, $x_{1,k}^{A*}$, and $x_{2,k}^{A*}$;
    \item Dataframe S (subscription): $\alpha$, $v$, $c$, $k$, $y_k^{S*}$, $\pi_k^{S*}$, $x_{1,k}^{S*}$, and $x_{2,k}^{S*}$;
    \item Dataframe P (social planner): $\alpha$, $v$, $c$, $k$, $y_k^{P*}$, $W_k^{*}$, $x_{1,k}^{P*}$, and $x_{2,k}^{P*}$.
\end{itemize}
With these tables, we can numerically show all the claims in the Propositions in our main text.


\section{Preliminary Empirical Evidence}\label{sec:empirical}

In this appendix, we provide some preliminary empirical evidence for the results generated from our theoretical analysis. Since content moderation is an increasingly important topic getting attention from users, managers, and policy makers, content moderation policies of social media platforms are ever-evolving and dynamic. Therefore, it is difficult to collect a data set that is comprehensive of all platforms. In this analysis, we rely on a list of 103 social media platforms composed by \textit{Influencer Marketing Hub}.\footnote{``101+ Social Media Sites You Need to Know in 2021.'' \url{https://influencermarketinghub.com/social-media-sites/}.} 

For each social media platform, we collected the texts of their content moderation policy and also  information about their major revenue models. The former was copied from a platform's community guidelines or terms of use. The latter was found by searching ``[\textit{name of platform}] business model'' or ``how does [\textit{name of platform}] make money'' on Google and referring to the relevant search results. 
Among the 103 platforms, we focus on those that are published in English language and precluded instant messaging platforms such as WhatsApp since they do not fit our modeling context. We also excluded platforms whose content moderation policy information could not be found and/or revenue models are not reported or ambiguous. This reduced the number of social media platforms we analyze to 67.

We hired independent graders from Mechanical Turk to read and decode the text of content moderation policy of each platform. Specifically, each grader was asked to give answers to a set of yes/no questions (given in Table \ref{tab: q_for_graders}) after reading the entire text of a platform's content moderation policy. The first question (Q1) asks whether a platform moderates content at all. Questions Q2-Q10 ask if the platform moderates particular types of content potentially offensive to some users\footnote{These categories were summarized by the researchers after reading the content moderation policies of around 30 platforms.}.
\begin{table}[h]
    \centering
    \caption{Questions for graders}
    \begin{tabular}{c l}
    \hline
    & Does [\textit{name of platform}]... \\
    \hline
    Moderation or not: & Q1: ... remove any content posted by users? \\
    Content categories: & Q2: ... remove sexual/adult content such as nudity?\\
    & Q3: ... remove illegal content such as terrorism, drug, arm selling, and etc.?\\
    & Q4: ... remove hate content toward a group (based on race, gender, sexual \\
    & \quad\quad orientation, and etc.)?\\
    & Q5: ... remove content related to harassment/bullying/threats?\\
    & Q6: ... remove content related to violence/blood/injury?\\
    & Q7: ... remove spam or repeated content?\\
    & Q8: ... remove content that violates others' privacy?\\
    & Q9: ... remove promotional or self-promotional content?\\
    & Q10: ... remove misleading information such as fake pictures, news, and etc.?\\
    \hline
    \end{tabular}
    \label{tab: q_for_graders}
\end{table}

At least five graders were assigned to each platform's content moderation policy. An attention-check question saying ``regardless of the true answer, check `No' for this question'' was also included in the survey. We precluded the responses which missed this attention-check question. To incentivize graders to give quality answers, we also gave bonus to graders if at least 80\% of  their responses were consistent with the rest of the graders. 

In Propositions \ref{prop:ad_eq} and \ref{prop:subsc_eq}, we claimed that a platform under advertising is more likely to conduct content moderation and when conducting content moderation, a platform under subscription does so more aggressively. We show some preliminary evidence of these results based on our data.

Among the 67 platforms, only two  platforms do not conduct content moderation (based on the majority response to Q1) and they both adopt subscription as revenue models.
For the remaining 65 platforms, we count how many of the 9 categories of content each platform moderates by aggregating the graders' responses. We use two different ways of aggregating questions Q2-Q10:
\begin{enumerate}
    \item \textit{Percentage (PER)}: the score a platform gets for Q$i$ ($i=2,3,...,10$) is the share of graders who respond ``yes'' to this question.
    \item \textit{Majority rule (MR)}: the score a platform gets for Q$i$ ($i=2,3,...,10$) is 1 if at least $50\%$ of the graders respond ``yes'' to this question, and 0 otherwise.
\end{enumerate}
Then we sum up all the scores a platform gets across questions Q2-Q10, which gives the number of content categories each platform moderates, denoted as $categories_{PER}$ or $categories_{MR}$. A higher $categories_{PER}$ or $categories_{MR}$ indicates a stricter content moderation policy. We also define a dummy variable $AD$ for each platform which takes the value 1 if the platform's major revenue source is advertising, and 0 if subscription. 

We run the following regressions across the 65 platforms which conducts content moderation:
\begin{align}
    categories_{PER} = \beta_0 + \beta_1 AD + \epsilon,\\
    categories_{MR} = \beta'_0 + \beta'_1 AD + \epsilon'.
\end{align}
Based on our theoretical results, a platform under subscription moderates content more aggressively than one under advertising given that it moderates content, so we expect the signs of $\beta_1$ and $\beta_1'$ to be negative. The regression results are shown in Table \ref{tab:reg}.

\begin{table}[htp] \centering 
  \caption{Regression results} 
  \label{tab:reg} 
\begin{tabular}{@{\extracolsep{5pt}}lcc} 
\\[-1.8ex]\hline 
\hline \\[-1.8ex] 
 & \multicolumn{2}{c}{\textit{Dependent variable:}} \\ 
\cline{2-3} 
\\[-1.8ex] & $categories_{PER}$ & $categories_{MR}$ \\ 
\\[-1.8ex] & (1) & (2)\\ 
\hline \\[-1.8ex] 
 $AD$ & $-$0.235 & $-$0.751$^{**}$ \\ 
  & (0.294) & (0.365) \\ 
  & & \\ 
 Constant & 7.365$^{**}$ & 8.516$^{**}$ \\ 
  & (0.213) & (0.264) \\ 
  & & \\ 
\hline \\[-1.8ex] 
Observations & 65 & 65 \\ 
R$^{2}$ & 0.010 & 0.063 \\ 
\hline 
\hline \\[-1.8ex] 
\textit{Note:}  & \multicolumn{2}{r}{ $^{**}$p$<$0.05} \\ 
\end{tabular} 
\end{table} 

We see that the directions of estimated $\beta_1$ and $\beta_1'$ are as expected, which provides preliminary support for our theoretical predictions. The estimate based on the majority rule aggregation is significant and the one based on percentage aggregation is in the same direction but less precisely estimated. We anticipate the $categories_{PER}$ measure to be noisier than $categories_{MR}$ since the latter focuses on the response on which the graders reach a consensus and the former does not require that. This would explain  $categories_{PER}$ measure is less precisely estimated.


\end{document}